\newcommand{\fullVersion}{def}

\newif\ifcomm
\commfalse

\newif\ifacm
\acmtrue

\newif\ifacmart
\ifacm   
   \acmarttrue
\else
   \acmartfalse
\fi

\newif\ifs
\sfalse

\documentclass[sigconf]{acmart}

\newcommand{\SSSInstance}{y}
\newcommand{\queueOfOverflows}{b}
\newcommand{\sumOfOverflows}{{B}}

\newcommand{\frameOffset}{M}

\newcommand{\ceil}[1]{ \left\lceil{#1}\right\rceil}

\newcommand{\parentheses}[1]{ \left({#1}\right)}

\newcommand{\inc}[1]{$#1 = #1 + 1$}
\newcommand{\dec}[1]{$#1 = #1 - 1$}

\newcommand{\smallMultError}[1]{(1+o(1))}

\newcommand{\window}{W}

\newcommand{\numBlocks}{k}
\newcommand{\blockSize}{\frac{\window}{\numBlocks}}

\usepackage{hhline}

\usepackage{graphicx}

\usepackage{amssymb}
\usepackage{algpseudocode}
\usepackage[Algorithm]{algorithm}
\usepackage{graphicx}
\usepackage{multirow}
\usepackage{amsmath}
\usepackage{amsthm}
\usepackage{amssymb}
\usepackage{algorithm}
\usepackage{subfig}
\usepackage{dblfloatfix}
\usepackage{url}
\usepackage{algpseudocode}
\usepackage{times}
\usepackage[normalem]{ulem}
\usepackage{todonotes}
\usepackage{xspace}
\usepackage{balance}
\algtext*{EndFunction}
\algtext*{EndIf}
\algtext*{EndFor}
\algtext*{EndWhile}
\algtext*{EndIf}
\usepackage{hyperref}
\definecolor{darkred}{rgb}{0.7,0,0}
\definecolor{darkgreen}{rgb}{0,0.5,0}
\hypersetup{colorlinks=true,
        linkcolor=darkred,
        citecolor=darkgreen}
\ifacm 

\else 
    \usepackage[cmex10]{amsmath}
    \usepackage{amsthm}  
    \newtheoremstyle{boldthm}{}{}{\itshape}{}{\bfseries}{.}{ }{\thmname{#1}\thmnumber{ #2}\thmnote{ (#3)}} 
    \theoremstyle{boldthm}
\fi

\ifacmart
\else
  \newtheorem{theorem}{Theorem}
  \newtheorem{definition}{Definition}
  \newtheorem{lemma}{Lemma}
  \newtheorem{Obs}{Observation}
  \newtheorem{corollary}{Corollary}
  
\fi
 
\ifcomm
    \newcommand{\mycomm}[3]{{\color{#2} \textbf{[#1: #3]}}}
    \newcommand{\Fmycomm}[3]{\footnote{{{\color{#2} \textbf{[#1: #3]}}} }}
\else
    \newcommand{\mycomm}[3]{}
    \newcommand{\Fmycomm}[3]{}    
\fi
\newcommand{\IK}[1]{\mycomm{IK}{red}{#1}}

\newcommand{\changed}[1]{\textcolor{blue}{#1}}

\ifs
    \ifacmart
        \usepackage[small,bf]{caption}
    \else
           \usepackage[hang,small,bf]{caption}
    \fi
\else
    
\fi

\def\compactify{\itemsep=0pt \topsep=0pt \partopsep=0pt \parsep=0pt}
  \let\latexusecounter=\usecounter

\ifacm 
  \newcommand{\bp}{\begin{proof}}
  \newcommand{\bpo}{ \begin{proof}[Proof Outline] }
  \newcommand{\ep}{\end{proof}}       
\else 
  \newcommand{\bp}{\begin{IEEEproof}}     
  \newcommand{\bpo}{ \begin{IEEEproof}[Proof Outline] }
  \newcommand{\ep}{\end{IEEEproof}}       
\fi

\newcommand{\be}{\begin{equation}}
\newcommand{\ee}{\end{equation}}

\newcommand{\NB}{\psi}

\newcommand{\NBound}{{Z_{1 - \frac{{{\delta }}}{2}}}HW^{-1}{\varepsilon_s}^{ - 2}}

\newcommand\MyIncludeGraphics[2][]{
    \IfFileExists{#2}{%
        \includegraphics[#1]{#2}%
    }{%
        \missingfigure[figwidth=7.0cm]{Missing #2}%
    }%
}%

\newcommand{\T}[1]{\smallskip\noindent\textbf{#1}} 

\newcommand{\set}[1]{\left\{#1\right\}}         

\newcommand{\eps}{\varepsilon}

\newcommand{\newVar}[2]{\newcommand{#1}{\ensuremath{#2}\xspace}}
  \newVar{\server}{S}
  \newVar{\client}{C}
  \newVar{\rclient}{R_c}
  \newVar{\rserver}{R_s}

\providecommand{\vs}{vs. }
\providecommand{\ie}{\emph{i.e.,} }
\providecommand{\eg}{\emph{e.g.,} }

\newcommand{\HHAWAII}{Memento\xspace}
\newcommand{\HHHAWAII}{H\text{-}Memento\xspace}
\newcommand{\name}{Memento\xspace}
\newcommand{\Hname}{H\text{-}Memento\xspace}

\begin{document}
    \title{
        Memento: Making Sliding Windows Efficient \mbox{for Heavy Hitters}
    }

  \ifacm 
    \ifacmart
    \newcommand{\aut}[2]{#1\texorpdfstring{$^{#2}$}{(#2)}} 
\author{
  \aut{Ran Ben Basat}{1,2},\, 
  \aut{Gil Einziger}{3,4},\, 
  \aut{Isaac Keslassy}{2,5},\, 
  \aut{Ariel Orda}{2},\, 
  \aut{Shay Vargaftik}{2},\, 
  \aut{Erez Waisbard}{4}
}%
      \affiliation{
$^1$ \textit{Harvard University} \quad
$^2$ \textit{Technion}\quad  
$^3$ \textit{Ben Gurion University} \quad
$^4$ \textit{Nokia Bell Labs} \quad 
$^5$ \textit{VMware}
}
          \renewcommand{\shortauthors}{R. Ben Basat \textit{et al.}} 
     \fi

  \fi 
\begin{abstract}

\IK{
5. In the extended version, I am not sure why, but all figures are postponed to the end and missing.}

Cloud operators require real-time identification of Heavy Hitters (HH) and Hierarchical Heavy Hitters (HHH)  for applications such as load balancing, traffic engineering, and attack mitigation. 
However, existing techniques are slow in \mbox{detecting new heavy hitters.}


In this paper, we make the case for identifying heavy hitters through \textit{sliding windows}. \sout{We demonstrate that sliding windows detect heavy hitters quicker and more accurately than current methods.}
\changed{Sliding windows detect heavy hitters quicker and more accurately than current methods, but to date had no practical algorithms.}
Accordingly, we introduce, design and analyze the \textit{Memento} family of \changed{sliding window} algorithms for the HH and HHH problems in the single-device and network-wide settings. Using extensive evaluations, we show that our single-device solutions attain similar accuracy and are by up to $273\times$ faster than existing window-based techniques.
Furthermore, we exemplify our network-wide HHH detection capabilities on a realistic testbed. To that end, we implemented  Memento as an open-source extension to the popular HAProxy cloud load-balancer.
In our evaluations, using an HTTP flood by 50 subnets, our network-wide approach detected the new subnets faster, and reduced the number of undetected flood requests by up to $37\times$ compared to the alternatives. 

\end{abstract}
\ifdefined\fullVersion
\renewcommand{\changed}[1]{{#1}}
\renewcommand{\sout}[1]{}
	\settopmatter{printfolios=true,printacmref=false} 
    \setcopyright{none}
	\renewcommand\footnotetextcopyrightpermission[1]{} 
	\pagestyle{plain} 
\else


\copyrightyear{2018} 
\acmYear{2018} 
\setcopyright{acmcopyright}
\acmConference[CoNEXT '18]{The 14th International Conference on emerging Networking EXperiments and Technologies}{December 4--7, 2018}{Heraklion, Greece}
\acmPrice{15.00}
\acmDOI{10.1145/3281411.3281427}
\acmISBN{978-1-4503-6080-7/18/12}

\renewcommand{\textrightarrow}{$\to$}
\fi
\maketitle


\ifacm 
    \sloppypar
\else 
\fi


\vspace{-.2cm}
\section{Introduction}

Cloud operators require fast and accurate single-device and network-wide detection of \emph{Heavy Hitters (HH)} (most frequent flows) and of \emph{Hierarchical Heavy Hitters (HHH)} (most frequent subnets) to attain real-time visibility of their traffic. 
These capabilities are essential building blocks in network functions such as load balancing~\cite{LBSigComm,LoadBalancing,SilkRoad,LBConext}, traffic engineering~\cite{TrafficEngeneering,TrafficEngneering2} and attack mitigation~\cite{DDoSwithHHH,DDOSwithHHH3,HHHSwitch,DDoSwithHHH2,PseudoWindowHHH}.

Quickly identifying changes in the HH and HHH is a key challenge~\cite{Change1} and can have a dramatic impact on the performance of such application. 
For example, faster detection of HH flows allows load-balancing and traffic engineering solutions to respond to traffic spikes swiftly. For attack mitigation systems, quicker and more accurate detection of HHH subnets means that less attack traffic reaches the victim. 
This is particularly important for combating \emph{Distributed Denial of Service (DDoS)} attacks on cloud services, as they become a growing concern with the increasing number of connected devices (\ie Internet of things)~\cite{DynAttack,DDoSReport}. 

In this work, we show that \emph{sliding windows} are faster than interval based measurements in detecting new (hierarchical) heavy hitters.
Unfortunately, the idea of using a sliding window for HHH was previously dismissed, as the existing sliding-window algorithms were ``markedly slower and less space-efficient in practice", to quote~\cite{HHHMitzenmacherArXiv}. 
\changed{Intuitively, this is because the sampling methods used for accelerating interval methods do not naturally extend to sliding windows, even for the simpler HH problem.}
As a result, the merits of sliding windows have not been properly evaluated. Moreover, sliding windows do not provide network-wide measurement capabilities, as opposed to interval approaches~\cite{HHTagging, Sigcomm2018Networkwide, RexfordNetworkwide,Everflow}. Consequently, most applications that use  HH or HHH rely on interval based measurements~\cite{DDoSwithHHH2,DDOSwithHHH3,DDOSwithHHH4,DDoSwithHHH}. 

\T{Contributions.}
Our goal is to make sliding windows practical for network applications. Accordingly, we focus on the fundamental HH and HHH problems in both single-device and network-wide measurement scenarios. We then introduce the \textit{Memento} family of four algorithms---one for each problem (\ie HH and HHH) and each measurement scenario (\ie single-device and network-wide). These are rigorously analyzed and provide worst-case accuracy guarantees. Moreover, in the network-wide setting, we maximize the accuracy guarantee given a per-packet control bandwidth~budget. 

Using extensive evaluations on real packet traces, we show that the Memento algorithms achieve speedups of up to 14$\times$ in HH and up to 273$\times$ in HHH when compared to existing sliding-window solutions. We also show that they match the speed of the fastest interval-based algorithm~\cite{RHHH}. Our algorithms detect emerging (hierarchical) heavy hitters consistently faster than interval-based approaches, and their accuracy is similar to that of slower sliding-window solutions.  

Next, we implement a proof-of-concept network-wide HH and HHH measurement system. The controller uses our network-wide algorithms, and the measurement points are implemented on top of the popular HAProxy cloud load-balancer, which we extended with capabilities to rate-limit traffic from entire subnets. 
We evaluate the achievable accuracy given a per-packet bandwidth budget for reporting measurement data to the control. We introduce new communication methods and compare them with a traditional approach. We create an HTTP flood attack from 50 subnets and show that the detection time is near-optimal while using a bandwidth budget of 1 byte per packet. For the same budget, our methods exhibit a reduction of up to 37$\times$ in the number of undetected flood requests compared to the alternative.
Finally, we \sout{(anonymously)} open-source the Memento algorithms and the HAProxy cloud \mbox{load-balancer extension~\cite{TRCode}.}



\changed{
\section{Background}
Streaming algorithms~\cite{muthukrishnan2005data} are designed to process a stream (sequence) of elements (in our case, packets) while analyzing the underlying data. 
The main challenge of these algorithms is the sheer volume of the data that they are required to process, motivating space-efficient solutions that \mbox{process elements extremely fast.}
}

\changed{One of the most studied streaming problems is that of identifying the \textit{Heavy Hitters} (HH) (\ie \textit{elephant flows}) -- the elements that frequently appear in the data.
For instance, Space Saving (SS)~\cite{SpaceSavings} is a popular HH algorithm. 
SS utilizes a set of $m$ counters, each associated with a key (flow identifier) and a value. Whenever a packet arrives, SS increments the value of its flow's counter if it exists. Otherwise, if there is a free counter, it is allocated to the flow with a value of $1$. Finally, if no available counter exists, SS replaces the identifier of the flow with the smallest value with 
that of the current flow and increments its value. 
For example, assume that the minimal counter is associated with flow $x$ and has a value of $4$, while flow $y$ does not have a counter. If a packet of flow $y$ arrives, we will reallocate $x$'s counter to $y$ and set its value to $5$, leaving $x$ without a counter. 
When queried for the value of flow $z$, SS returns the value of the counter associated with $z$, or the value of the minimal counter if there is no counter for $z$.
}

\changed{
SS runs on \emph{intervals}, \ie it estimates the flow sizes from the beginning of the measurement and is often reset to allow its data to be fresh~\cite{DDOSwithHHH3}. 
Another way for analyzing only the recent data is to use a sliding window algorithm~\cite{DatarGIM02} in which an answer to a query only reflects the last $W$ packets. WCSS~\cite{WCSS} extends Space Saving 
to sliding windows, and achieves constant update and query time. Unfortunately, WCSS is too slow to keep up with line rates, and it
\ifdefined\fullVersion
We expand on WCSS and how it is generalized by Memento in Section~\ref{sec:HHAWAII}.
\else
serves as a baseline in this work. We expand on WCSS in the full version of the paper~\cite{fullVersion}. 
\fi
\IK{please add a placeholder reference! Ran: I prefer not, there's no way I'll submit a paper with a missing reference, but may miss that we haven't replaced a dummy identifier.}
}

\changed{
\textit{Hierarchical Heavy Hitters} (HHH) are a generalization of the HHH problem in which we identify frequent IP \emph{networks}. That is, rather than looking for the large flows, we look for the networks whose aggregated volume exceeds a predetermined threshold. To that end, MST~\cite{HHHMitzenmacher} proposed to utilize SS for tracking all networks. Specifically, it uses one SS instance for each network size and whenever a packet arrives, it computes all its prefixes (specifically $H$ updates, as the size of the hierarchy) 
MST has two main drawbacks: first, it makes multiple  SS updates, while even a single update may be too slow for keeping up with line rates; second, it solves the problem on intervals rather than on sliding 
windows.
To alleviate the first problem, RHHH~\cite{RHHH} proposes to draw a single random integer $i$ uniformly between $1$ and $V$ (for a parameter $V\ge H$). If $i\le H$ then RHHH makes a \emph{single} SS update to the $i$'th prefix, and otherwise it ignores the packet. For example, for the above packet, $i=7$ would be ignored while $i=3$ would lead the algorithm to feed $181.0.0.0$ into the relevant SS instance.
While RHHH is fast enough to keep up with the line rate, its approach does not seem to extend to sliding windows easily, a gap we close in this paper.
}

\section{Why sliding windows?}\label{sec:2} 

In this paper, we argue that once a new heavy hitter emerges, the sliding window method identifies it most quickly and accurately. Therefore, network applications that capitalize on sliding windows can potentially react faster to changes in the traffic.
For simplicity, we consider accurate measurements, but the results are also valid for approximate measurements.

\T{Window \vs interval.} 
We start by comparing \textit{sliding windows} to the \emph {Interval} method that is commonly used in HHH-based DDoS mitigation systems~\cite{DDoSwithHHH2,DDOSwithHHH3,DDoSwithHHH,PseudoWindowHHH}.
As depicted in Figure~\ref{fig:whyWindow2}, the Interval method relies on sequential interval measurements. Usually, the measurement data is available at the end of each measurement interval, wherein the \emph{improved} Interval method it is accessible throughout each measurement period. 
There are two possible failure modes, namely: failing to detect a new heavy hitter (false negative), or falsely declaring a heavy hitter (false positive). Algorithms that follow the (improved) Interval method would either have false positives or false negatives. In contrast, sliding windows can avoid both errors.  
To show this, we start with the following definitions:
\begin{definition}[Window Frequency]
We denote by $f^{W}_x$ the \textit{window frequency} of flow $x$, \ie the \emph{number} of packets transmitted by $x$ over the last $W$ packets.
\end{definition}
\begin{definition}[Normalized Window Frequency]
We denote by $\frac{f_x^W}{W}$ the \textit{normalized window frequency} of flow $x$, \ie the \emph{fraction} of  $x$'s packets within the last $W$ packets.
\end{definition}
Next, window heavy hitters are flows whose normalized window frequency is larger than a user-defined threshold:
\begin{definition}[Window Heavy Hitter]
Flow $x$ is a \textit{window heavy hitter} if its normalized window frequency  $\frac{f_x^W}{W}$ is larger than $\theta$, where $\theta \in (0,1)$ is a user-defined threshold. 
\end{definition}

\smallskip
\T{Window optimality.} The optimal detection point for new window heavy hitters is simply once their normalized window frequency is above a user-defined threshold. Reporting a flow earlier is \emph{wrong} (false positive), and reporting it afterwards is \emph{(too) late}.
%
This means that sliding window measurements, \emph{by definition}, have an optimal detection time. 
%

\smallskip
\T{Motivation.} 
We motivate the definition for window heavy hitters with an experimental scenario where a new flow appears during the measurement and consumes, \changed{at a constant rate,} a larger-than-the-threshold portion of the traffic after its initial appearance. We measure how long it takes for each measurement method to identify the new heavy hitter and evaluate the \mbox{following measurement methods:}

(i) Interval. The window frequency of each flow is estimated at the end of every measurement. 
This method represents limitations of sampling techniques~(\eg \cite{BUS,RHHH}) that require time to converge and thus cannot provide estimates during the measurement.
{(ii) Improved interval.} Same as \textit{interval}, but flow frequencies are estimated upon the arrival of each packet. This represents the best case scenario for the Interval method. 
(iii) Window. Sliding window, where frequencies are estimated upon packet arrivals. 


\begin{figure}[t!]
\subfloat[\label{fig:whyWindow2}An example of the periodic \textit{interval} and \textit{sliding window} methods. In this scenario, consider a threshold of nine packets. The solid-green flow is a window heavy-hitter as it has ten packets within the sliding window. However, the measurement interval method does not detect the green flow, as it only has five packets within the current interval (false negative). Intuitively, one can identify the green flow by lowering the threshold to 4 packets, but in that case, the striped red flow is detected as well (false positive).
] {\includegraphics[width=\columnwidth]
{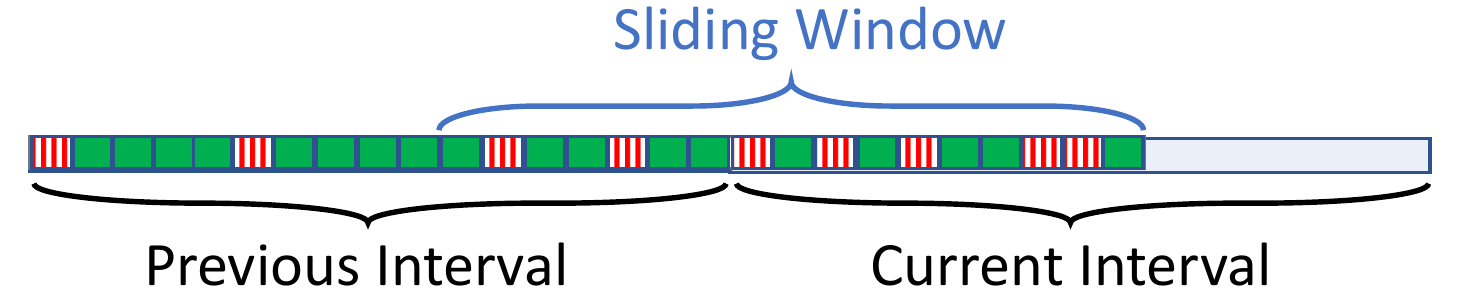}}\\
\subfloat[\label{fig:whyWindow} Effect of a new heavy hitter's frequency on its detection time. The x-axis is the ratio of the normalized heavy hitter's frequency and the user-defined threshold. The y-axis is the expected detection time in windows. For instance, when the frequency is twice the threshold, it takes a window algorithm half a window to detect the new heavy hitter whereas interval-based algorithms require \mbox{between 0.6-1.0 windows.}
] {\includegraphics[width=\columnwidth]
{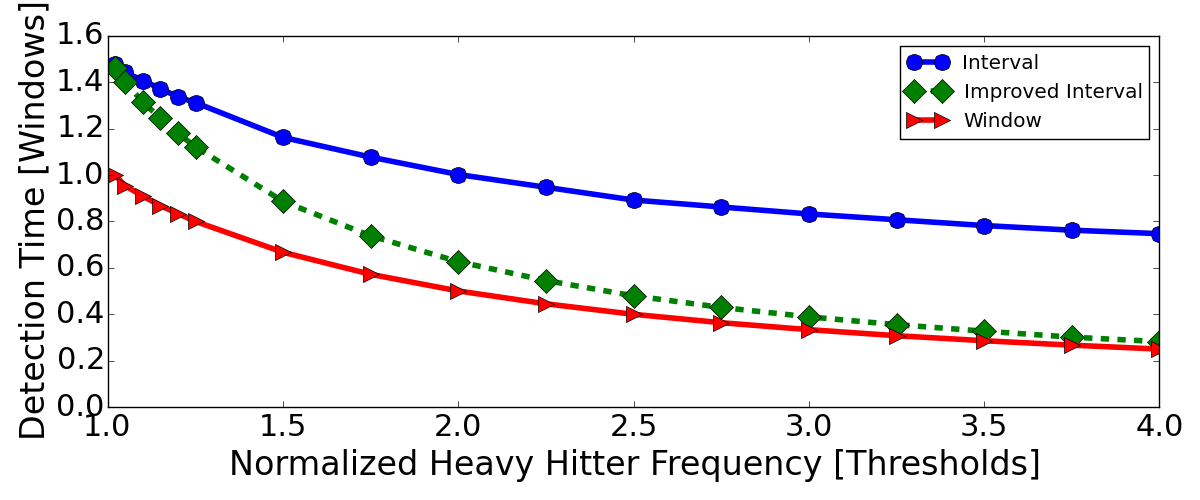}}
\caption{Sliding windows compared to intervals.}\vspace*{-0.4cm}
\end{figure}

Figure~\ref{fig:whyWindow} plots the detection time for each method as a function of the normalized frequency of the new heavy hitter. Intuitively, larger heavy hitters are detected faster, because less time passes before their normalized window frequency reaches the threshold.  Indeed, the \textit{sliding window} approach is always faster than the \textit{Interval} and \textit{improved Interval} methods. When the frequency is close to the detection threshold, we get up to $40\%$ faster detection time compared to the Interval method.  At the end of the tested range, sliding windows are still over $5\%$ quicker. The Interval method is the slowest, as it estimates frequencies only at the end of the measurement. Thus, such a usage pattern is undesired for systems such as load balancing and attack mitigation. 


\section{Sliding window algorithms}\label{sec:algorithms}
Our next step is to make sliding windows accessible to cloud operators. We do so by first introducing new single-device algorithms that are significantly faster than existing techniques, and then extend them to efficient network-wide algorithms that combine information from many measurement points to obtain a global perspective. 

\subsection{Heavy Hitters on Sliding Windows}
\label{sec:HHAWAII}

\T{Natural approach.} Our goal is to produce faster sliding window algorithms.
Intuitively, one can accelerate the performance of a heavy hitter algorithm by sub-sampling the packets. That is, we would like to sample packets with a probability of $\tau$, use an HH algorithm with a window size of $W\cdot \tau$ packets, and then multiply its estimations by a factor of $\tau^{-1}$. Unfortunately, this does not yield the desired outcome as the number of samples from the window varies whereas sliding-window HH algorithms are designed for fixed-sized windows. 
Specifically, since the actual number of samples from the $W$ sized window is distributed $\text{Bin}(W,\tau)$, this approach results in an additional error of  
\small$\pm \Theta\left(\sqrt{W(1-\tau)\cdot \tau^{-1}}\right)$
\normalsize~in the size of the reference window.  
Since we are interested in small values of $\tau$ to achieve speedup (see Section~\ref{sec:eval}), this \mbox{approach results in a considerable error.}

\T{\HHAWAII{} overview.}
The key idea behind \textit{\name} is to decouple the computationally expensive operation of updating a packet (\emph{Full update}) from the lightweight operation of \emph{Window update}.
Specifically, for each packet, \name{} performs the Full update operation with probability $\tau$; otherwise, it makes the quicker Window update.

Therefore, \name alternates between the fast {Window updates} and the slower {Full updates}. 
Full updates include (1) forgetting outdated data \emph{and} (2) adding a new item to the measurement data structure. On the other hand,  {Window updates} only involve (1) forgetting outdated data. 
That is, \name{} maintains a $W$-sized window but most of the packets within that window are missing. 
Thus, it attains speedup but avoids the additional error that is caused by uniform samples. The concept is exemplified in Figure~\ref{fig:DPI}. 



\T{Implementation.}
\label{sec:existing}
For simplicity, we built \HHAWAII{} on top of an existing sliding window HH algorithm. This makes it easier to implement, verify, and then compare with the current approaches. We picked WCSS as the underlying algorithm~\cite{WCSS}, but our approach is general and works on other window algorithms as well (\eg~\cite{FAST,HungAndTing}).
\changed{Intuitively, when $\tau=1$, \HHAWAII{} becomes identical to WCSS as it performs a full update for each packet.}
\begin{figure}[t!]\centering
    \subfloat[Our HH algorithm, \HHAWAII{},  utilizes two update methods: a slow Full update, and a faster Window update that only updates the sliding window.  Speedup is achieved by performing Full updates for a small fraction of the packets. Here, \changed{the flow id} $4$ is inserted following the coin flip and $2$ leaves the window (regardless \mbox{of the coin flip).} \label{fig:DPI}]{\includegraphics[width = 0.9\columnwidth
]{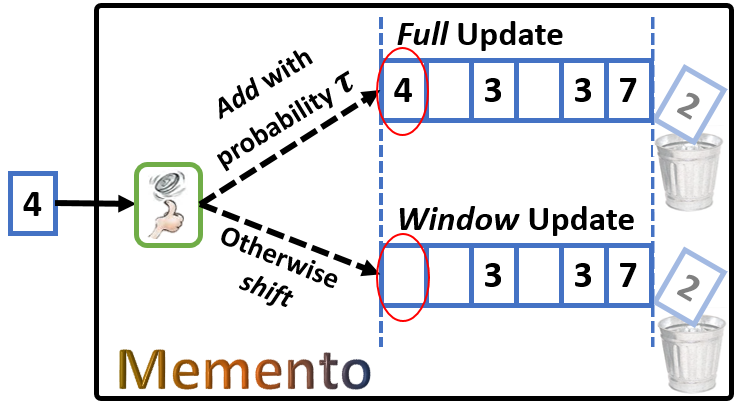}}\\\vspace*{-0.2cm}
    \subfloat[Our HHH algorithm, \HHHAWAII{}, simply updates \HHAWAII{} with a single random prefix, achieving constant time~complexity.\label{fig:VMI}]{\includegraphics[width = 0.9\columnwidth]{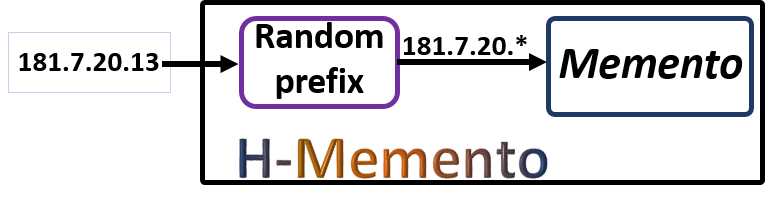}}
\caption{High-level overview of our algorithms. 
}
        \label{fig:contribution}
\end{figure}
\ifdefined\fullVersion
As detailed in Algorithm~\ref{alg:hhawaii}, given some error parameter $\epsilon_a$ such that $W\cdot \epsilon_a\gg1$, \HHAWAII{} divides the stream into \emph{frames} of size $\window$, where each frame is then further partitioned into $\numBlocks\triangleq\ceil{\frac{4}{\epsilon}}$  equal-sized \emph{blocks}. Intuitively, \HHAWAII{} keeps count of how many times each item arrived during the last frame, and each time this counter reaches a multiple of the block size, it records this  event as an \emph{overflow}.
\HHAWAII{} uses a queue of queues $\queueOfOverflows$, which contains $\numBlocks+1$ queues -- one queue for each block that overlaps with the current window. Each queue in $\queueOfOverflows$ contains an ordered list of items that overflowed in the corresponding block.
When a block ends, we remove the oldest queue from $\queueOfOverflows$, as it no longer overlaps with the window. Additionally, we append a new empty queue to $\queueOfOverflows$. Note that \HHAWAII{} does not count accurately, but instead uses a Space Saving~\cite{SpaceSavings} instance (denoted $\SSSInstance$) to approximately count the in-frame frequency.
Space Saving (SS) is an algorithm that uses \emph{counters} to provide frequency estimations and to find the heavy hitters over a stream or interval. Allocated with $k=\ceil{4/\epsilon}$ counters, it guarantees that the additive error is bounded by $W/k\le W\epsilon/4$ (when the number of packets is $W$, as in \HHAWAII{}). We show that despite the approximate count within frames, \HHAWAII{} keeps the overall error bounded as guaranteed.
Intuitively, \HHAWAII{} provides the machinery to both speed-up Space Saving and extend it to sliding windows.
%
Finally, $\SSSInstance$ is cleared at the end of each frame.

The frequency of an item is estimated by multiplying its number of overflows by the block size and adding the remainder of its appearance count as reported by $\SSSInstance$. In \cite{HHHMitzenmacher}, MST has one-sided error, and thus we choose to keep the error one-sided as well for comparison purposes. To do so, \HHAWAII{} adds $2\blockSize$ to each item's estimation. It then multiplies the result by $\tau^{-1}$ as a Full update is performed on average once per $\tau^{-1}$ packets.  The table $\sumOfOverflows$ counts the number of overflows for each item for quick frequency queries. \sout{Finally,} \HHAWAII{} de-amortizes the update of $\sumOfOverflows[x]$, \mbox{achieving constant worst case time.} 
\changed{To that end, when processing a packet, \HHHAWAII{} pops (at most) one flow from the queue of the oldest block (see lines~\ref{line:empty-tail}-\ref{line:remove-from-tail}). This ensures that the worst case update time is constant as we are guaranteed that by the end of the block we will have an empty queue and $B$ will be fully updated.
Finally, for finding the heavy hitters themselves (rather than just estimating flow sizes), \HHAWAII{} iterates over the flows with entry in $B$ and estimate their sizes. Since every heavy hitter must overflow in the window, we are guaranteed that it will have such an entry.}
\small
\begin{algorithm}[!t]
	\caption{\HHAWAII$(W,\epsilon,\tau)$}\label{alg:hhawaii}
	\begin{algorithmic}[1]
		\State Initialization: $\numBlocks=\ceil{4/\epsilon}, \SSSInstance=\mathit{SpaceSaving}.\mathit{init}(\numBlocks), \frameOffset = 0, \sumOfOverflows=\mbox{Empty hash},$
		$\queueOfOverflows = \mbox{Queue of $\numBlocks+1$ empty queues}$.
		\Function{WindowUpdate()}{}
		\State $\inc\frameOffset \mod W$
		\If {$\frameOffset = 0$} 
		$y$.\Call {flush}{\null}\label{line:flush}\Comment new frame
		\EndIf
		\If {$\frameOffset\mod\blockSize = 0$}
		\Comment new block
		\State$\queueOfOverflows$.\Call {pop}{\null}\label{line:pop}
		\State$\queueOfOverflows$.\Call {append}{new empty queue}
		\EndIf
		\If {$\queueOfOverflows$.tail is not empty} \Comment remove oldest item \label{line:empty-tail}
		\State $oldID=\queueOfOverflows$.tail.\Call {pop}{\null}\label{line:pop-from-b}
		\State \dec{\sumOfOverflows[oldID]}
		\If {$\sumOfOverflows[oldID]=0$}
		 $\sumOfOverflows$.\Call{remove}{$oldID$}\label{line:remove-from-tail}
		\EndIf
		\EndIf
		\EndFunction
		\Function{FullUpdate}{Item x}
		\State $WindowUpdate()$
		\State$\SSSInstance.$\Call{add}{$x$}  \Comment add item
		\If {$\SSSInstance$.\Call{query}{$x$}$= 0\ mod\ \blockSize$} \Comment{overflow}
		\State $\queueOfOverflows$.head.\Call{push}{$x$}\label{line:push}
		\If {$\sumOfOverflows$.\Call{contains}{x}\}}
		\inc{\sumOfOverflows[x]}
		\Else $\ \sumOfOverflows[x]=1$ \Comment{Adding $x$ to $\sumOfOverflows$}
		\EndIf
		\EndIf
		\EndFunction
		\Function{Update}{Item x}
		\If {$\mathit{Uniform}(0,1) \le \tau$}
		 \Call{FullUpdate}{$x$}
		\Else
		    \Call{WindowUpdate}{$x$}
		\EndIf
		\EndFunction
		\Function{Query}{Item x}
		\If {$\sumOfOverflows$.\Call{Contains}{$x$}}
		\State\Return {$\tau^{-1}\cdot\left(\blockSize\cdot\parentheses{\sumOfOverflows[x]+2} + \parentheses{\SSSInstance.\Call{query}{x} \mod \blockSize}\right)$}
		\Else\, \Return $ \tau^{-1}\cdot\left(2\blockSize+y.\Call{query}{x}\right)$ \Comment {no overflows}
		\EndIf
		\EndFunction
	\end{algorithmic}
\end{algorithm}
\normalsize
\else
\changed{We provide more background on the algorithmic implementation of WCSS and Memento in the full version of the paper~\cite{fullVersion}. \IK{see my comment, please add a reference placeholder}}
\fi

\subsection{Extending to Hierarchical Heavy Hitters}
\label{sec:HHHAWAII}
Hierarchical heavy hitters monitor subnets and flow aggregates in addition to individual flows. We start by introducing existing approaches for HHH measurements on \mbox{sliding windows.} 

\T{Existing approaches.} 
In MST~\cite{HHHMitzenmacher},
multiple HH instances are used to solve the HHH problem. This design trivially extends to sliding windows by replacing the HH building blocks with window algorithms (e.g., WCSS~\cite{WCSS}). This was proposed by~\cite{HHHMitzenmacher} but dismissed as impractical. 
Replacing the underlying algorithms with \HHAWAII{} is slightly better as we can perform Window updates to most instances.
Unfortunately, the update complexity remains $\Omega(H)$ which may still be too slow. In contrast, \HHHAWAII{} achieves constant time updates, matching the complexity of interval algorithms~\cite{RHHH}.
Another natural approach comes from the RHHH~\cite{RHHH} algorithm. RHHH shares the same data structure as MST but randomly updates at most a single HH instance which allows for constant time updates. 
\changed{Additionally, it makes small changes to the query procedure to compensate for the sampling error and guarantees that (with high probability) it will have no false negatives.}
This method does not work for sliding windows, as each HH instance is updated a varying number of times and monitors a possibly different window. 

\T{\HHHAWAII{}'s overview.}
In \HHHAWAII{} we differ from the lattice structure of RHHH and MST. That is, we maintain a single large \HHAWAII{} instance and use it to monitor all the sampled prefixes.  Therefore, we use just one sliding window to measure all subnets, which the underlying \HHAWAII{} does in constant time.  This approach also has engineering benefits such as code reuse, simplicity, and maintainability. 
The update procedure of \HHHAWAII{} is illustrated in Figure~\ref{fig:VMI}.
Next, we proceed with notations and definitions for the HHH problem, which we later use to detail \HHHAWAII{}.


 \small
\begin{table}[!t]
    \centering
    \begin{tabular}{|c|l|}
        \hline
        Symbol & Meaning \tabularnewline
        \hline
        $\mathbb{S}$ & The packet stream. \tabularnewline
        \hline
        $N$ & Current number of packets (the stream length).  \tabularnewline
        \hline
        $W$ & The window size.  \tabularnewline
        \hline
        $H$ & Size of the hierarchy.  \tabularnewline
                \hline
        $\tau$ & Sampling probability.   \tabularnewline
        \hline
        $V$ & Sampling ratio for HHH, $V\triangleq \frac{H}{\tau}$. \tabularnewline
        \hline
        $S^i_x$ & Variable for the i'th appearance of a prefix $x$.
        \tabularnewline
        \hline
        $S_{x}$ & Sampled prefixes with id $x$.
        \tabularnewline
        \hline
        $S$ &  Sampled prefixes from all ids.
        \tabularnewline

        \hline
        $\mathcal U$ & Domain of fully specified items. \tabularnewline
        \hline
        $\epsilon,\epsilon_s,\epsilon_a $ & Overall, sample, algorithm's error guarantee.  \tabularnewline
        \hline
        $\delta,\delta_s , \delta_a$ &   Overall, sample, algorithm's confidence. \tabularnewline
        \hline
        $\theta$ & Threshold parameter.    \tabularnewline
        \hline
        $C_{q|P}$ & Conditioned frequency of $q$ with respect to  $P$. \tabularnewline
        \hline
        $G(q|P)$ & Subset of $P$ with the closest prefixes to q.  \tabularnewline
        \hline
        $f_q$ & Frequency of prefix q \tabularnewline
        \hline
        $\widehat{f^{+}_q},\widehat{f^{-}_q}$ & Upper and lower bounds for $f_q$.  \tabularnewline
        \hline
        $Z$ &   inverse CDF of the normal distribution . \tabularnewline
        \hline
        $\mathcal B$ &  per-packet control bandwidth budget.  \tabularnewline
        \hline
        $\mathbb{O}$ &   the minimal header size (in bytes),  \tabularnewline
        \hline
        $E$ &   bytes required to report a packet.  \tabularnewline
        \hline
        $m$ &   number of measurement points.  \tabularnewline
        \hline
        $b$ &   number of samples in each report.  \tabularnewline
        \hline
        $\mathfrak E_b $ &  overall error in network-wide settings. \tabularnewline
        \hline
    \end{tabular}
    \caption{Summary of notations}
    \label{tbl:notations}
    \vspace{-0.8cm}
\end{table}
\normalsize
\T{HHH notations and definitions.}
\label{sec:HHHNotations}
For brevity, Table~\ref{tbl:notations} summarizes the notations used in this work. We consider IP prefixes (\eg $181.*$).  A prefix without a wildcard (e.g., $181.7.20.6$) is called \emph{fully specified}. The notation $\mathcal U$ is the domain of the fully specified items. 
A prefix $p_1$ \emph{generalizes} another prefix $p_2$ if $p_1$ is a prefix of $p_2$. For example, $181.7.20.*$ and $181.7.*$ generalize the (fully specified) $181.7.20.6$.
The \emph{parent} of a prefix is the longest generalizing prefix, \eg $181.7.*$ is $181.7.20.*$'s parent.  \mbox{Definition~\ref{def:generalization} formalizes this concept.}
\begin{definition}[Generalization]
    \label{def:generalization}
    Let $p,q$ be prefixes. We say that $p$ \textit{generalizes} $q$ and denote
    $p \preceq q$ if for each dimension $i$, $p_i = q_i$ or $p_i\preceq q_i$. 
    We denote the set of fully specified items generalized by $p$ using $H_p\triangleq \{e\in \mathcal U\mid e\preceq p\}$.  Similarly, the set of every fully specified item that is generalized by a set of prefixes $P$ is denoted by: $H_P\triangleq \cup_{p\in P} H_p$.
    Moreover, denote $p\prec q$  if $p \preceq q$ and $p\neq q$.
\end{definition}

Definition~\ref{def:generalization} also deals with the more general multidimensional case. For example, we can consider tuples of the form (source IP, destination IP). In that case, fully specified ``prefixes'' are fully determined in both dimensions,
\eg $(\langle181.7.20.6\rangle, \langle 208.67.222.222\rangle)$.
Also, observe that ``prefixes'' now have two parents, \eg
$(\langle181.7.20.*\rangle, \langle 208.67.222.222\rangle)$
and
$(\langle181.7.20.6\rangle, \langle 208.67.222.*\rangle)$
are both parents to\\
$(\langle181.7.20.6\rangle, \langle 208.67.222.222\rangle)$. 

The \textit{size of the hierarchy ($H$)}  is the number of different prefixes that generalize a fully specified prefix.
Next, we look at a set of prefixes $P$ and denote $G(p|P)$ as the set of prefixes in $P$ that are most closely generalized by the prefix $p$. That is, $G(p|P) \triangleq     \left\{ {h:h \in P,h \prec p,\nexists\,h' \in P\,\,s.t.\,h \prec h' \prec p} \right\}$.

For example, consider the prefix $p=<142.14.*>$ and the set $P = \left\{ {<142.14.13.*>,<142.14.13.14>} \right\}$, then we have $G(p|P) = \left\{<142.14.13.*>\right\}$.
The window frequency of a prefix $p$ is the total sum of packets within the window that are generalized by $p$, \ie  $f^W_p \triangleq  \sum\nolimits_{e \in H_p} {f^W_e} .$ Note that each packet is generalized by $H$ different prefixes. 
This motivates us to look at the \emph{conditioned} (residual) frequency that a prefix $p$ adds to a set of already selected prefixes $P$. The {conditioned frequency} is defined as:
$C_{p\mid P}\triangleq\allowbreak \sum_{e\in H_{P\cup\{p\}} \setminus H_P} f_e$.

We denote by $X^W_p$ the number of times prefix $p$ is sampled in the window, $\widehat{X^W_p}^{+}$ is an upper bound on $X^W_p$ and $\widehat{X^W_p}^{-}$ is a lower bound.  The notation $V \triangleq \frac{H}{\tau}$ stands for the sampling rate of each specific prefix.
We define:\\
$\widehat{f^W_p}\triangleq \widehat{X^W_p} V$ -- an estimator for $p$'s frequency.\\
$\widehat{f^W_p}^{+}\triangleq \widehat{X^W_p}^{+} V$ -- an upper bound for $p$'s frequency.\\
$\widehat{f^W_p}^{-}\triangleq \widehat{X^W_p}^{-} V$ -- a lower bound for $p$'s frequency.

We now define the \emph{depth} of a prefix (or a prefix tuple). Fully specified items are of depth 0, their parents are of depth 1 and more generally, the parent of an item with depth $x$ is of depth $x+1$.
 $L$ denotes the maximal depth; observe that this may be lower than $H$ (\eg in 2D byte-hierarchies $H=25$ and $L=9$). Hierarchical heavy hitters are calculated by iterating over all fully specified items (depth $0$). If their frequency is larger than a threshold of $\theta W$, we add them to the set $HHH_0$. Then, we go over all the items with depth $1$ and if their \emph{conditioned} frequency, with regard to $HHH_0$, is above $\theta W$, we \textbf{add} them to the set. We name the resulting set $HHH_1$ and repeat the process $L$ times, until the set $HHH_L$ contains the (exact) hierarchical heavy hitters. 
Unfortunately, we need space that is linear in the stream size to calculate exact HHH (and even plain heavy hitters)~\cite{TCS-002}.  Hence, as done by previous work~\cite{RHHH,HHHMitzenmacher,Cormode2003,Cormode2004,CormodeHHH}, \mbox{we solve approximate HHH.}


A solution to the approximate HHH problem is a set of prefixes that satisfies 
the \emph{Accuracy} and \emph{Coverage} conditions (Definition~\ref{def:deltaapproxHHH}).  Here, Accuracy
means that the estimated frequency of each prefix is within acceptable error bounds and Coverage means that the conditioned frequency of prefixes not included in the set is below the threshold. 
This does not mean that the conditioned frequency of prefixes that are included in the set is above the threshold. Thus, the set may contain a small \mbox{number of subnets misidentified as HHH (false positives).}
\begin{definition}[Approximate HHHs]
    \label{def:deltaapproxHHH}
     An algorithm $\mathbb A$ solves {\sc {$(\delta,\epsilon, \theta)$ - Approximate Window Hierarchical Heavy Hitters}} if it returns a set of prefixes $P$ that, for an arbitrary run of the algorithm, satisfies: 
     
\textbf{Accuracy:} If $p\in P$ then \small$\Pr \left( {\left| {{f^W_p} - \widehat {{f^W_p}}} \right| \le \varepsilon W} \right) \ge 1 - \delta.$\normalsize 

\textbf{Coverage:} If $q \notin P$ then \small$\Pr \left( {{C_{q|P}} < \theta W} \right) \ge 1 - \delta .$\normalsize
\end{definition}



\begin{algorithm}[t]
    \small
    \begin{algorithmic}[1]
\Statex Initialization:  $\HHAWAII.\mathit{init}(H\cdot \epsilon_a^{-1}, W,\tau\cdot H)$
        \Function{Update}{ $x$}
        \State $Memento.update(RandomPrefix(x))$
        \EndFunction

        \Function{Output}{$\theta$}
        \State $HHH = \phi$
        \For{Level $\ell = 0$ up to $L$ }
        \For{ each $p$ in level $\ell$} \Comment{Only over  prefixes with a counter.}
        \State \label{line:cp}$\widehat{C_{p|HHH}} = \widehat{f^W_p}^{+} + calcPred(p,HHH) $
        \State \label{line:accSample}$\widehat{C_{p|HHH}} = \widehat{C_{p|HHH}}+ 2{{Z_{1 - {\delta}}}\sqrt {VW} }$\Comment{\scriptsize{} \mbox{Compensate for sampling}}
        \If {$\widehat{C_{p|HHH}}\ge \theta N$}
        $ HHH = HHH \cup \{p\}$ 
        \EndIf
        \EndFor
        \EndFor
        \State\Return $HHH$
        \EndFunction

    \end{algorithmic}
\normalsize
    \caption{\HHHAWAII{} $(W, \eps_a, \tau)$ }
    \label{alg:Skipper}
\end{algorithm}

\begin{algorithm}[t]
    \small
    \begin{algorithmic}[1]
        \Function{calcPred}{prefix $p$, set $P$}
        \State $R = 0$
        \For{ each $h\in G(p|P)$}
         \label{alg:first}$R = R - \widehat{f^W_h}^{-}$
        \EndFor
        \State \Return $R$
        \EndFunction
    \end{algorithmic}
    \normalsize
    \caption{calcPred for one dimension }
    \label{alg:randHHH}
\end{algorithm}

\begin{algorithm}[t]
    \small
    \begin{algorithmic}[1]
        \Function{calcPred}{prefix $p$, set $P$}
        \State $R = 0$
        \For{ each $h\in G(p|P)$}
        \label{alg:second}$R = R - \widehat{f_h}^{-}$
        \EndFor
        \For{ each pair $h,h'\in G(p|P)$}
        \State $q=glb(h,h')$
        \If {$\not\exists h_3 \neq h,h'\in G(p|P), q \preceq h_3$}
        \label{alg:third}$R = R + \widehat{f_q}^{+}$
        \EndIf
        \EndFor
        \State \Return $R$
        \EndFunction
    \end{algorithmic}
    \normalsize
    \caption{calcPred for two dimensions}
    \label{alg:randHHH2D}
\end{algorithm}

\T{\HHHAWAII{}'s full description.}
\label{accurateDescription} A pseudo-code for \HHHAWAII{} is given in Algorithm~\ref{alg:Skipper}. 
The output method performs the HHH set calculation as explained 
for exact HHH. The calculation yields an approximate result as we only have an approximation for the frequency of each prefix. Thus, we conservatively estimate conditioned frequencies. 

 For two dimensions, we use the inclusion-exclusion~principle (Definition~\ref{def:glb}) to avoid double counting.
\begin{definition}
    \label{def:glb}
    Denote by $glb(h,h')$ the greatest lower bound of $h$ and $h'$.
    $glb(h,h')$ is a unique common descendant of $h$ and $h'$ s.t. $\forall p : \left(q\preceq p\right)\wedge \left(p \preceq h\right)\wedge \left(p \preceq h'\right) \Rightarrow p = q.$
    If $h$ and $h'$ have no common descendants, $glb(h,h') =0$
\end{definition}

A pseudo code for the update method is given in Algorithm~\ref{alg:Skipper}, which is the same for one and two dimensions. The difference between these is encapsulated in the {\sc calcPred} method which uses Algorithm~\ref{alg:randHHH} for one dimension and Algorithm~\ref{alg:randHHH2D} for two.
In two dimensions, $C_{p|HHH}$ is first set in Line~\ref{line:cp} of Algorithm~\ref{alg:Skipper}. 
Then, we remove previously selected descendant heavy hitters (Line~\ref{alg:second}, Algorithm~\ref{alg:randHHH2D}) and finally we add back the common descendants (Line~\ref{alg:third}, Algorithm~\ref{alg:randHHH2D})). The sampling error is accounted for in Line~\ref{line:accSample}. 
Intuitively, our analysis shows which $\tau$ values guarantee that \HHHAWAII{} solves the approximate HHHs problem.
\ifdefined\fullVersion
A formal proof of the algorithm's correctness appears in Section~\ref{sec:analysis}.
\else

\fi

\newcommand{\thirdFigWidth}{5.12cm}

\subsection{Network-Wide Measurements}
\label{sec:distributed}
As Figure~\ref{fig:paperInAFig} illustrates, we now discuss a centralized controller that receives data from multiple clients and forms a network-wide view of the traffic (\eg network-wide HH or HHH). Similarly to~\cite{RexfordNetworkwide, HHTagging} we assume that there are several measurement points and that each packet is measured once. Our design focus is on two critical aspects of this system: (1) a \textit{communication method} between the clients and the controller that conveys the information in a timely and bandwidth-efficient manner, and (2) a fast \textit{controller algorithm}.


\T{Formal model.} 
First, we define a sliding window in the network-wide model as the last $W$ packets that were measured somewhere in the network. 
Intuitively, we want the controller to analyze the traffic of the most recent $W$ packets \emph{in the entire network}, as observed by the measurement points. For example, we may want to monitor the last million packets in the entire network.

\T{(1) Communication method.}
We now suggest three methods to communicate with the controller. For each method, the frequency of messaging with the control is according to the bandwidth budget ($\mathcal B$). That is, smaller reports can be sent more frequently but also deliver less information.  

\T{Aggregation.}
Existing HH algorithms are often \emph{mergeable}, \ie the content of two HH instances can be efficiently merged~\cite{AndersonIMSUM}. We are unaware of previous work that targets HHH, but since MST~\cite{HHHMitzenmacher} and RHHH~\cite{RHHH} use HH algorithms as building blocks then they can be merged as well. This capability motivates the \emph{Aggregation} communication method. In this method, each client periodically transmits all the entries of its HH algorithm to the controller. Given enough bandwidth, this method is intuitively the most communication-efficient, as all data is transmitted. 
However, as each message is large, we infrequently send messages to meet the bandwidth budget, which creates inaccuracies. 

\T{Sample.}
Most network devices are capable of transmitting uniform packet samples to the controller. Motivated by this capability, the Sample method samples packets with a fixed probability $\tau$, and sends a report to the controller once per $\tau^{-1}$ packets. Thus on average, each message contains a single sample. This information is enveloped by the usual packet headers that are required to deliver the packet in the network.  We observe that this uses a significant portion of the bandwidth for the header fields of the transmitted packet. Yet, the Sample method is considerably easier to deploy than the Aggregate option, as the nodes only sample packets and do not run the measurement algorithms. The communication pattern is network-friendly as we get a stable flow of traffic from the clients to the controller.   


\T{Batch.}
The Batch approach is designed to utilize bandwidth more efficiently than the Sample.
The idea is simple: instead of transmitting,  on average, a single sample per message, we send on average $b$ samples (\eg 100) per report.
That is, we send a report once per $\tau^{-1}b$ packets, containing all the sampled packets within this period. This pattern utilizes the bandwidth more efficiently as the payload ratio of the message is considerably higher. However, it also creates delays in reporting new information to the controller. Our analysis is used to find the optimal batch size $b$ and minimize \mbox{the total error.} 

\begin{figure}[t]
{\includegraphics[width=1.03\columnwidth, height=4cm]{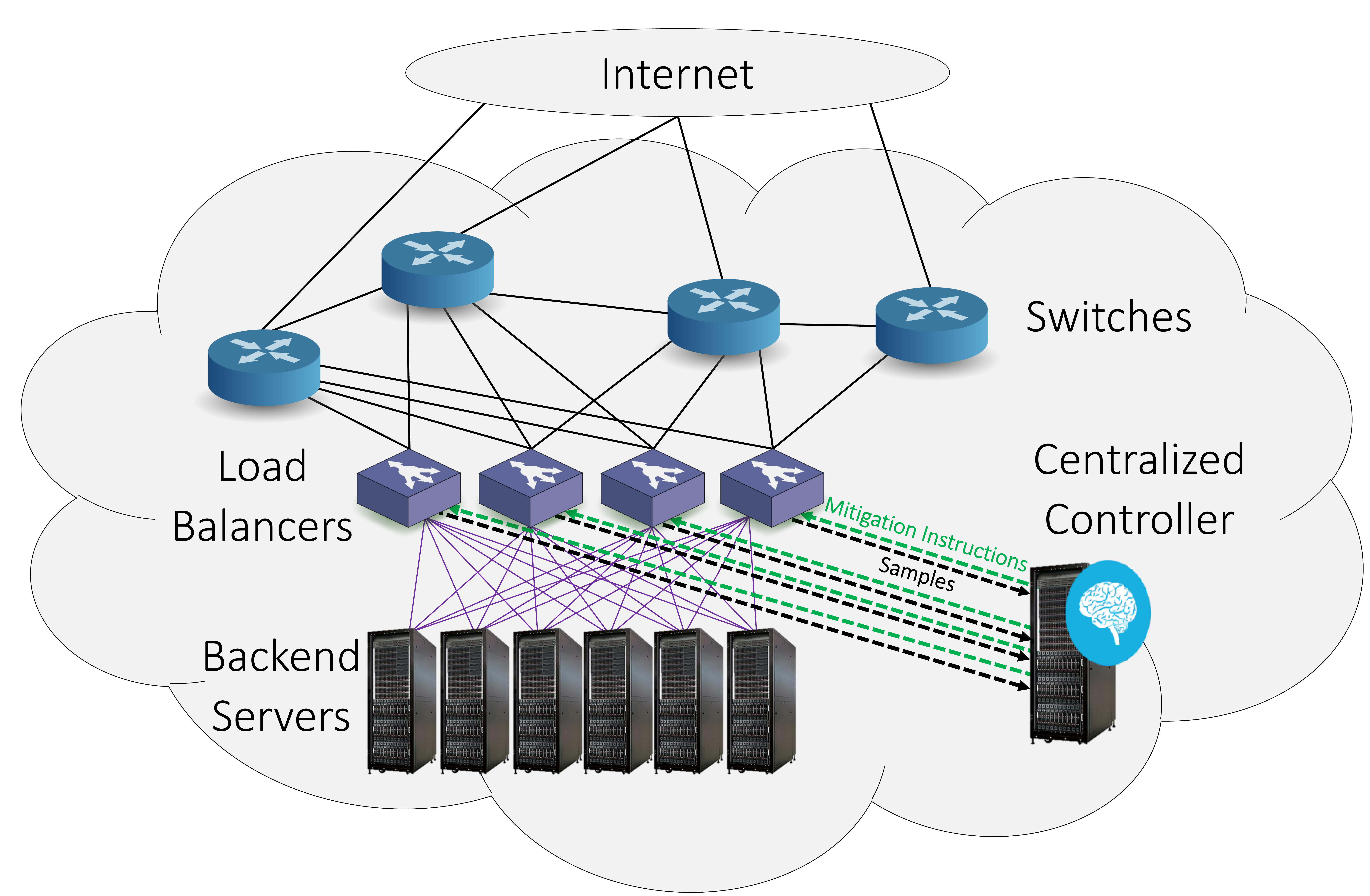}}
\caption{An overview of our system. The clients (load-balancers) perform the measurements and periodically send information to a centralized controller. The controller then runs a global sliding-window analysis. For example, in the case of an HTTP flood, it can mitigate the attack by instructing the clients which subnets to rate-limit or block.  
\label{fig:paperInAFig}}
\end{figure}

\T{(2) Controller algorithm.}
The controller maintains an instance of \HHAWAII{} or \HHHAWAII{} where we term the respective algorithms D-\HHAWAII{} and D-\HHHAWAII{}. The controller behaves slightly differently in each option.

\T{Aggregation.} Aggregation is used in this study only as a baseline. Thus, instead of implementing a specific algorithm, we simulate an \emph{idealized} aggregation technique with an unlimited space at the controller and no accuracy losses upon merging. 
As we later show, the Sample and Batch approaches outperform this Aggregation method; thus, we conclusively demonstrate that they are superior to \emph{any} aggregation technique. 

\T{Sample and Batch.} In the Sample and Batch schemes, the controller maintains a \HHAWAII{} or \HHHAWAII{} instance. When receiving a report, it first performs a Full update for each sampled packet and then makes Window updates for the un-sampled ones. 
In total, the Sample performs $\tau^{-1}$ updates and the Batch performs $\tau^{-1}b$ updates. 


\section{Analysis}
\label{sec:AnalOverviewAnalysis}
This section is divided into two parts; first, Section~\ref{anal:single} analyzes our single-device \HHAWAII{} and \HHHAWAII{} algorithms and shows accuracy guarantees. 
\ifdefined\fullVersion
\else
Due to space limits, some proofs are omitted from this version and are available in an online, anonymized,  full version of this paper~\cite{fullVersion}.
\fi
Next, Section~\ref{anal:networkwide} analyzes our network-wide D-\HHAWAII{} and D-\HHHAWAII{} algorithms, and explains how to find the optimal batch size (in terms of guaranteed error)  given a certain (per-packet) bandwidth budget. 

\subsection{\HHAWAII{} and \HHHAWAII{} Analysis}
\label{anal:single}


This section surveys the main theoretical results for \HHAWAII{} and for \HHHAWAII{}. These assure correctness as long as the sampling probability is large enough. 


\T{Formal model.} Our traffic is modeled as a stream $\mathbb{S}$. It is  initially empty, and a packet is added at each step.  A sliding-window algorithm considers only the last $W$ packets, denoted as  $\mathbb{S}^W$.
The notation $f^{W}_e$ denotes the frequency of flow $e$ in $\mathbb{S}^W$. Given $e$, a heavy hitters algorithm provides an estimator $\widehat{f^{W}_e}$ for $f^{W}_e$.
We formalize the problem as follows:

\begin{definition}
    \label{Def:probFE}
    An algorithm solves {\sc {$(\epsilon, \delta)$ - Window Frequency Estimation}} if given a query for a flow ($x$),
    it provides  $\widehat{f^W_{x}}$ such that
    $\Pr \left[ {\left| {{f^W_x} - \widehat {{f^W_x}}} \right| \le \varepsilon W} \right] \ge 1 - \delta.$
\end{definition}

\T{\name{}.} Theorem~\ref{th:main} is the main theoretical result for \name. It states that \name solves the {\sc {$(\epsilon, \delta)$ - Window Frequency Estimation}} problem  for $\varepsilon = \varepsilon_a + \varepsilon_s$ whenever it is allocated  $O(1/\epsilon_a)$ counters and has a sampling probability that satisfies $\tau\ge Z_{1-\frac{\delta}{4}}W^{-1}\epsilon_s^{-2}$, where
$Z$ is the inverse of the cumulative density function of the normal distribution with mean $0$ and standard deviation of $1$. Note that $Z_{1-\frac{\delta}{4}}$ satisfies $Z<4$ for any $\delta> 10^{-6}$.
In other words, {the theorem emphasizes the trade off between the amount of space allocated and the sampling rate, for achieving a target error bound $\epsilon$.} Specifically, if the algorithm has many counters (\ie have a low $\epsilon_a$), then we can afford a higher $\epsilon_s$ (\ie the sampling rate can be low). \ifdefined\fullVersion
The complete analysis for \HHAWAII{} and \HHHAWAII{} is presented in Appendix~\ref{sec:analysis}.
\else
\mbox{The proof appears in the full version of this paper~\cite{fullVersion}.}
\fi

 \begin{theorem}
 \small
     \label{th:main}
          \HHAWAII{} solves {\sc {$(\epsilon, \delta)$ - Windowed Frequency Estimation}} for $\varepsilon = \varepsilon_a + \varepsilon_s$ and $\tau\ge Z_{1-\frac{\delta}{4}}W^{-1}\epsilon_s^{-2}$.\normalsize
 \end{theorem}

 \T{\HHHAWAII{}.}
 Theorem~\ref{thm:correctness} is our main result for \HHHAWAII{}. It says that \HHHAWAII{} is correct for any $\tau> \NBound$, 
 where $H$ is the size of the hierarchic domain ($H=5$ for source hierarchies and $H=25$ for (source,destination) hierarchies). \ifdefined\fullVersion
The complete analysis for \HHAWAII{} and \HHHAWAII{} is presented in Appendix~\ref{sec:analysis}.
\else
The proof is given in the full version of this paper~\cite{fullVersion}.
\fi

\begin{theorem}
     \label{thm:correctness}
     \small
     \HHHAWAII{} solves {\sc {$(\delta,\epsilon, \theta)$ - Approximate Windowed Hierarchical Heavy Hitters}} for $\tau\ge\NBound$ .
     \normalsize
 \end{theorem}

\subsection{D-\HHAWAII{} and D-\HHHAWAII{} Analysis}
\label{anal:networkwide}
We now provide analysis for our {network-wide} D-\HHAWAII{} and D-\HHHAWAII{} algorithms. Intuitively, the error in these algorithms comes from two origins. First, an error due to \textit{sampling}, which is quantified by Theorems~\ref{th:main} and~\ref{thm:correctness}. However, there is an additional error that is caused by the \textit{delay in transmission}, as the measurement points only send the sampled packets once in every $b\tau^{-1}$ packets. 
If a measurement point has a low traffic rate, it may take long time for it to see $b\tau^{-1}$ packets; in this case, all of its samples may be obsolete and may not belong in the most recent window. Therefore, our first step is to reason about the accuracy impact of the \mbox{Batch and Sample methods.}

\T{Notations and definitions.}
We denote the \emph{bandwidth budget} as $\mathcal B$ bytes/packet. That is, $\mathcal B$ determines how much traffic is used for communicating between the measurement points and the controller. This communication is done using standard packets, which have header field overheads. We denote by $\mathbb{O}$ 
the minimal header size (in bytes) of the chosen transmission protocol (\eg 64 bytes for TCP). 
Next, reporting a sampled packet requires $E$ bytes (\eg 4 bytes for srcip, or 8 bytes for (srcip,dstip) pair). We also denote by $m$ the total number of measurement points. 

\T{Model.}
Intuitively, we can choose two (dependent) parameters: the sampling rate, $\tau$, and the batch size $b$. That is, each measurement point samples with probability $\tau$ until it gathers $b$ packets. At this point, it assembles an $(\mathbb{O}+Eb)$-sized packet that encodes the sampled packet and sends it to the controller. As the expected number of packets required to gather a $b$ sized batch is $b\tau^{-1}$, the bandwidth constraint can be written as 
$
(\mathbb{O}+Eb) / (b\tau^{-1}) \le \mathcal B.
$
Specifically, this allows to express the maximum allowed sampling probability as $\tau = \mathcal Bb / (\mathbb{O}+Eb)$ since
sampling at a lower rate would not utilize the entire bandwidth and would result \mbox{in sub-optimal accuracy.}

\T{Accuracy of the Batch and Sample methods.}
We can now quantify the error of the Batch and Sample methods. 
Intuitively, we have to factor the delays in communication (as we only report per a fixed number of packets to stay within the bandwidth budget).
For example, if there are two measurement points in which one processes a million requests per second while the other only a thousand, the batches of the second point would include many obsolete packets that are not within the current window. 
However, recall that these reports only reflect $b\tau^{-1}$ packets at each of the $m$ points.  Therefore, we conclude that: 
\begin{theorem}
\label{thm:batchdelay}
     The error created by the delayed reporting in the Batch method is bounded by $mb\tau^{-1}$.   
\end{theorem}

Next, Theorems~\ref{th:main} and~\ref{thm:correctness} enable us to bound the sampling error as a function of $\tau$, while Theorem~\ref{thm:batchdelay} bounds the delayed reporting error. The following theorem applies for  D-\HHAWAII{} (using $H=1$) and for D-\HHHAWAII{} (using the appropriate $H$ value). As the round trip time inside the data center is small compared to window sizes that are of interest, the error caused by the delay of packet transmissions is negligible, and thus we do not factor it here.
\changed{Theorem~\ref{thm:batch} quantifies the overall error in the Batch method; the error of the Sample method is derived when setting $b=1$. }
\begin{theorem}
\label{thm:batch}
Given overhead $\mathbb{O}$, batch size $b$,  bandwidth budget $\mathcal B$, 
sample payload size $S$, window size $W$ and confidence $\delta_s$, the overall error $\mathfrak E_b$ (in packets) is at most: 
$$\mathfrak E_b = m(\mathbb{O}+Eb) / \mathcal B+\sqrt{HWZ_{1-\frac{\delta_s}{2}}(\mathbb{O}+Eb) / (\mathcal Bb)}.$$
\end{theorem}
\begin{proof}
According to Theorem~\ref{thm:correctness}, we have that $\epsilon_s = \sqrt{HW^{-1}Z_{1-\frac{\delta_s}{2}}\tau^{-1}}$.
This means that our overall error is bounded~by:
\begin{multline*}
    \mathfrak E_b = bm\tau^{-1}+W\epsilon_s
    = bm\tau^{-1}+\sqrt{HWZ_{1-\frac{\delta_s}{2}}\tau^{-1}}\\
    = m(\mathbb{O}+Eb) / \mathcal B+\sqrt{HWZ_{1-\frac{\delta_s}{2}}(\mathbb{O}+Eb) / (\mathcal Bb)}.\qedhere
\end{multline*}
\end{proof}
Formally, we showed a bound of $\mathfrak E_b$, for each choice of $b$. The guarantees for the Sample method are given by fixing $b=1$. The next step is to use
Theorem~\ref{thm:batch} to calculate the optimal batch size $b$ given a bandwidth budget $\mathcal B$. Thus, we get the best achievable accuracy for the Batch method within the bandwidth limitation.\mbox{ We have:}
$$
\frac{\partial \mathfrak E_b}{\partial b} = mE/\mathcal B + \frac{HWZ_{1-\frac{\delta_s}{2}}(E/B - \mathbb{O}/(\mathcal Bb^2) )}{2\sqrt{(\mathbb{O}+Eb) / (\mathcal Bb)}}.
$$
We then compare this expression to zero to compute the optimal batch size $b$. This is easily done \mbox{with numerical methods.} 

For example, for a TCP connection ($\mathbb{O}=64$); ten measurement points ($k=10$); source IP hierarchy ($E=4, H=5$); error probability of $\delta=0.01\%$; a window of size $W=10^6$; and a bandwidth quota of $\mathcal B=1$ byte per packet, the optimal batch size is $b=44$. The resulting (overall) error guarantee is 13K packets (i.e., an error of $1.3\%$).
Increasing the bandwidth budget 
to $\mathcal B=5$ bytes decreases the absolute error to 5.3K packets ($0.53\%$) while increasing the optimal batch size to $b=68$.
When increasing the window size ($W$), the \emph{absolute} error increases by an $O\left(\sqrt{W}\right)$ factor and the error (as a fraction of $W$) decreases. For example, increasing the window size to $10^7$ increases the optimal batch size to $b=109$, while reducing the error to $0.15\%$. Alternatively, 2D source/destination hierarchies (increasing $H$ from $5$ to $25$) result in a slightly larger error and \mbox{a higher optimal batch size.}

Figure~\ref{fig:distErrorGuarantee} illustrates the accuracy guarantee provided by each method. We compare three synchronization variants -- Sample, Batch with $b=100$, and Batch with an optimal $b$ (varies with $\mathcal B$), as explained above.
As depicted, Sample has the smallest delay error and yet provides the worst guarantees as it conveys little information within the bandwidth budget.  
The 100-Batch method has lower a sampling error (as its sampling rate is higher), but its reporting delay makes the overall error larger. For larger values of $\mathcal B$, the optimal batch size grows closer to $100$ and the accuracy gap narrows. 

\begin{figure}[t]
{\includegraphics[width=1.0\columnwidth]{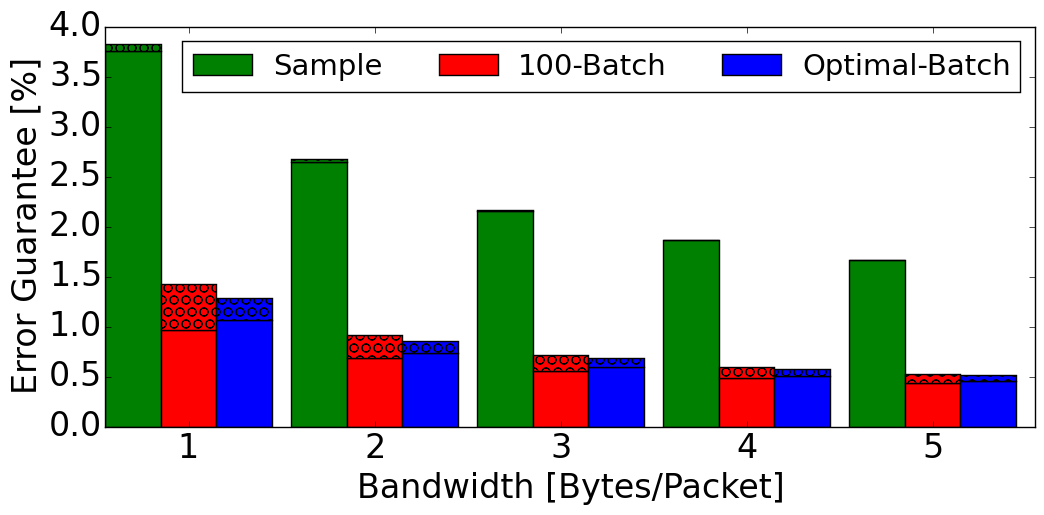}
\caption{Comparing the accuracy \emph{guarantees} of varying synchronization techniques. The parts hatched with circles quantify the bound on the error that is caused by \mbox{the delayed synchronization.}}
\label{fig:distErrorGuarantee}}
\end{figure}
\section{Evaluation}
\label{sec:eval_measure}

\T{Server.} Our evaluation was performed on a Dell 730 server running Ubuntu 16.04.01 release. The server has 128GB of RAM and an Intel Xeon CPU E5-2667 v4@ 3.20GHz.

\T{Traces.} We use real packet traces collected from
an edge router~\emph{(Edge})~\cite{UCLA}, a datacenter \emph{(Datacenter)}~\cite{Benson}, and a CAIDA backbone link \emph{(Backbone)}~\cite{CAIDACH16}.

\T{Algorithms and implementation.}
For the HH problem, we compare  \HHAWAII{} and  WCSS~\cite{WCSS}. For WCSS we use our \HHAWAII{} implementation without sampling ($\tau=1$).
For the HHH problem, we compare \HHHAWAII{} to MST~\cite{HHHMitzenmacher} and RHHH~\cite{RHHH} (interval algorithms). We use the code released by the original authors of these algorithms.
We also form the \emph{Baseline} sliding window algorithm by replacing the underlying algorithm in MST~\cite{HHHMitzenmacher} with WCSS.
Specifically, MST proposed to use Lee and Ting's algorithm~\cite{Lee:2006:SME:1142351.1142393} as WCSS was not known at the time. By replacing the algorithm with the WCSS, a state of the art window algorithm, we compare with the best variant known today.

\T{Yardsticks.} We consider source IP hierarchies in byte granularity ($H=5$) and two-dimensional source/destination hierarchies ($H=25$).
Such hierarchies are also used in~\cite{RHHH,HHHMitzenmacher,CormodeHHH}.
We run each data point $5$ times and use two-sided Student's t-tests to determine the $95\%$ confidence~intervals.

We evaluate the empirical error in the \emph{On Arrival model}~\cite{WCSS,RAP}, where for each packet we estimate its flow (denoted $s_t$) size. We then calculate the \emph{Root Mean Square Error (RMSE)}, \ie
$RMSE(Alg) \triangleq \sqrt{\frac{1}{|N|}\sum_{t=1}^{N}(\widehat{f_{s_t}} - f_{s_t})^2.}$

\subsection{Heavy Hitters Evaluation }
We evaluate the effect of the sampling probability $\tau$ on the operation speed and empirical accuracy of
\HHAWAII{}, and use the speed and accuracy of WCSS as a reference point for this evaluation.
The notation X-WCSS stands for WCSS that is allocated X counters (for X$\in\set{64, 512, 4096}$). Similarly, the notation X-\HHAWAII{} is for \HHAWAII{} with X counters. The window size is set to $W\triangleq5$ million packets and the interval length is set to $N\triangleq16$ million packets.

\begin{figure}[!t]
\includegraphics[width = \columnwidth]
                {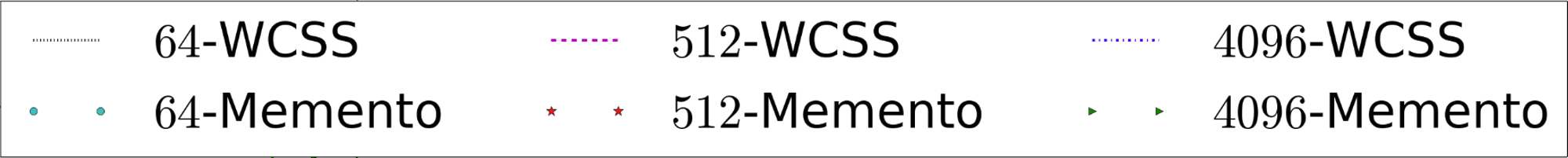}
\subfloat[\label{H/HHSpeedEdge}Edge Trace: Speed]{\includegraphics[width = 0.49\columnwidth]
            {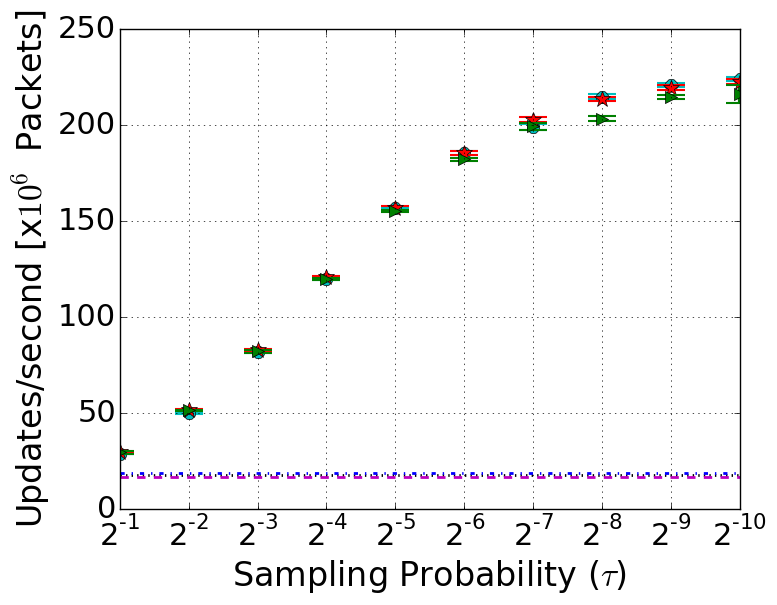}}
\subfloat[\label{HHAccEdge}Edge Trace: Error]{\includegraphics[width = 0.49\columnwidth]
            {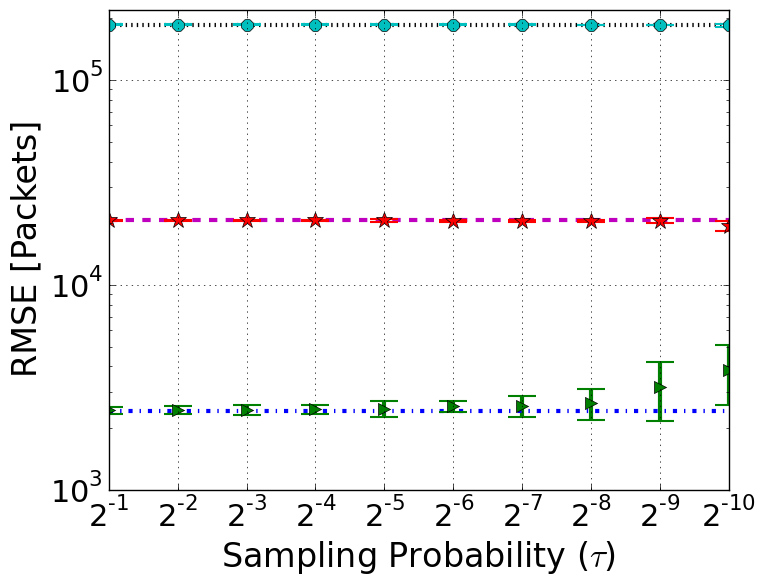}}\\
 \subfloat[\label{HHSpeedDC}Datacenter Trace:  Speed]{\includegraphics[width = 0.49\columnwidth]
            {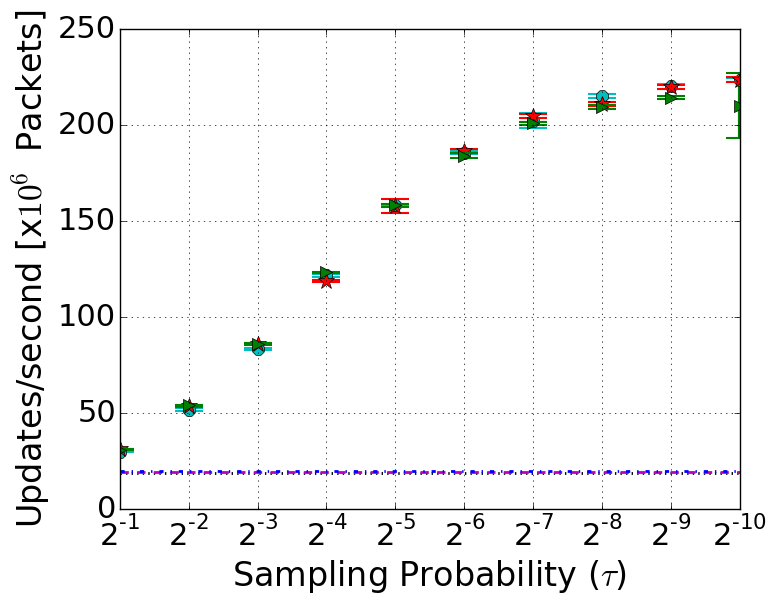}}
\subfloat[\label{HHAccDC}Datacenter Trace:  Error]{\includegraphics[width = 0.49\columnwidth]
            {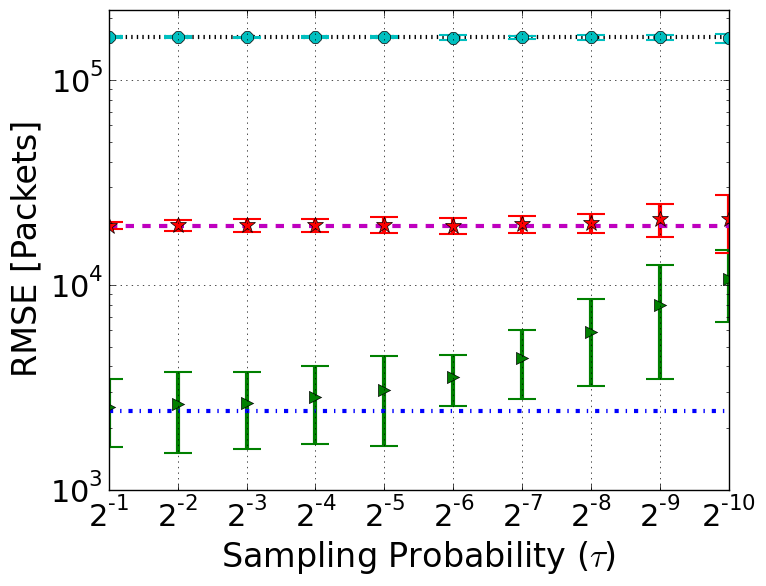}}\\
\subfloat[\label{HHSpeedCHI16}Backbone Trace:  Speed]{\includegraphics[width = 0.49\columnwidth]
            {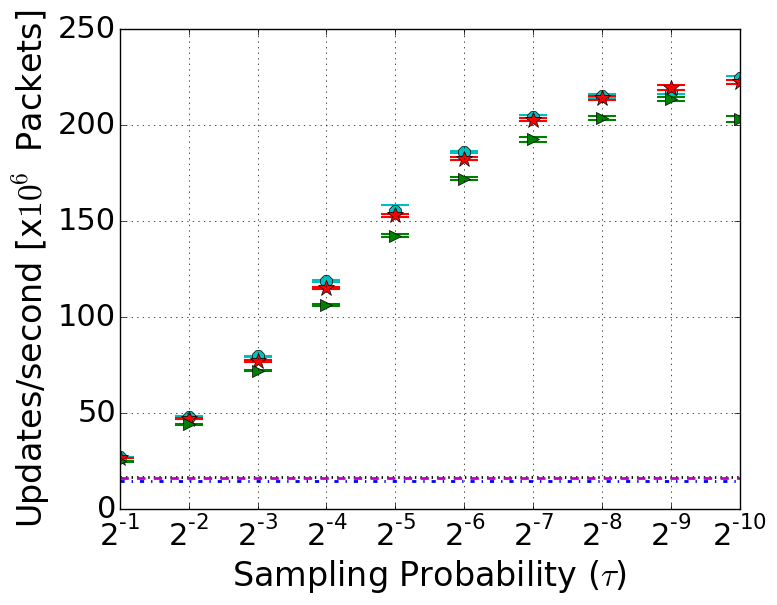}}
\subfloat[\label{HHAccCHI16 }Backbone Trace: Error]{\includegraphics[width = 0.49\columnwidth]
            {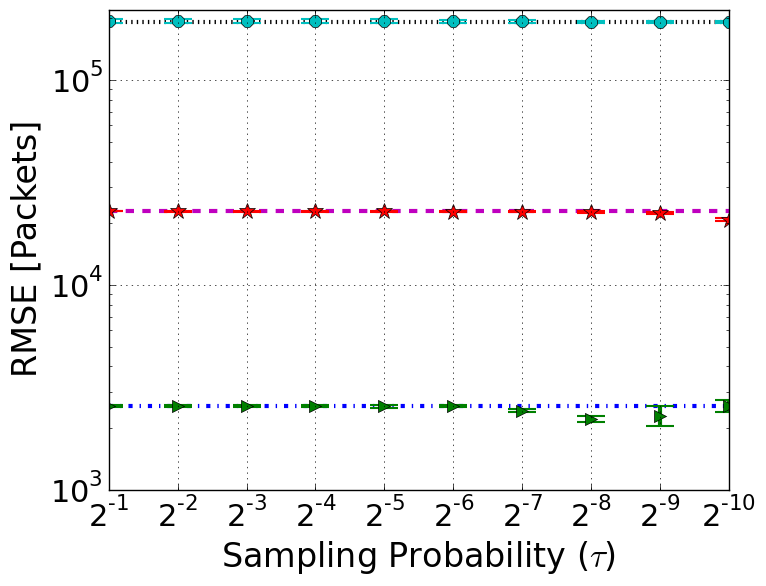}}
\caption{
\label{fig:tau}
Effect of the sampling probability $\tau$ on the speed and accuracy for varying number of counters given three different traces and a window size $W= 5M$. $\HHAWAII{}$ is considerably faster than WCSS, but the accuracy of both algorithms is almost the same despite $\HHAWAII{}'s$ use of sampling, except when the sampling rate is low and the number of counters is high. Even then, this is mainly evident in \mbox{the skewed Datacenter trace.}
}
\end{figure}
As depicted in Figure~\ref{fig:tau}, the update speed is determined by the sampling probability and is almost indifferent to changes in the number of counters. \HHAWAII{} achieves a speedup of up to 14$\times$ compared to WCSS.
As expected, allocating more counters also improves the accuracy. It is also evident that the error of \HHAWAII{} is almost identical to that of WCSS, which indicates that it works well for the range. The smallest evaluated $\tau$, namely, $\tau=2^{-10}$, already exhibits slight accuracy degradation, which shows the limit of our approach. It appears that a larger number of counters, or heavy tailed workloads (such as the Backbone trace), allow for even smaller sampling probabilities without a \mbox{impact to the attained accuracy.}

\subsection{Hierarchical Heavy Hitters Evaluation}

\T{\HHHAWAII{} \vs existing window algorithm.} Next, we evaluate \HHHAWAII{} and compare it to the Baseline algorithm.  We consider two common types of hierarchies, namely a one-dimensional source hierarchy ($H=5$) and two-dimensional source/destination hierarchies ($H=25$). Note that \HHHAWAII{} performs updates in constant time while the Baseline does it in $O(H)$.  Following the insights of Figure~\ref{fig:tau}, we evaluate \HHHAWAII{} with a sampling rate $\tau$ such that $\tau\ge H\cdot 2^{-10}$, so that each of the $H$ prefixes is sampled with a probability of at least $2^{-10}$. \changed{That is, we do not allow sampling probabilities of $\tau<H\cdot 2^{-10}$ to get an effective sampling rate of at least $2^{-10}$, which is the range in which Memento is accurate.}

We evaluate three configurations for each algorithm, with a varying number of counters. The notation 64H denotes the use of $64\cdot 5 = 320$ counters when $H=5$, and 1600 counters when $H=25$. The notations 512H and 4096H follow the same rule. In the Baseline algorithm, the counters are utilized in $H$ equally-sized WCSS instances, while \HHHAWAII{} has a single \HHAWAII{} instance with that many counters.

Figure~\ref{fig:HHHtau} shows the evaluation results. We can see how $\tau$ is the dominating performance parameter. 
\HHHAWAII{} achieves up to a $52\times$ speedup in source hierarchies and a $273\times$ speedup in source/destination hierarchies. This difference is explained by the fact that the Baseline algorithm makes $H$ expensive Full updates for each packet, while \HHHAWAII{} usually performs a \mbox{single Window update}.

 \begin{figure}[t]
 \includegraphics[width = \columnwidth, height =1cm]                 {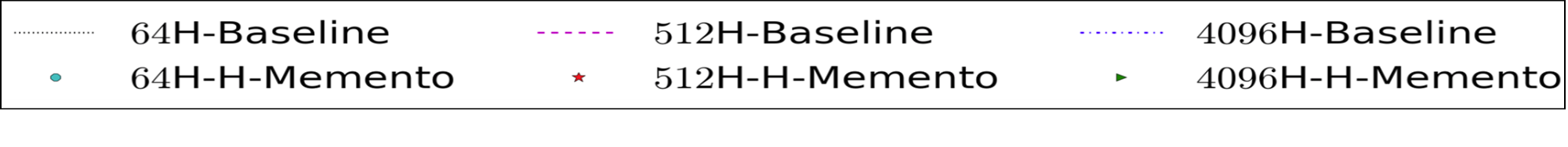}\vspace{-0.4cm} \\
\subfloat[\label{HHHSpeedCHI16}Backbone trace - 1 D (H=5)]{\includegraphics[width = 0.5\columnwidth]
            {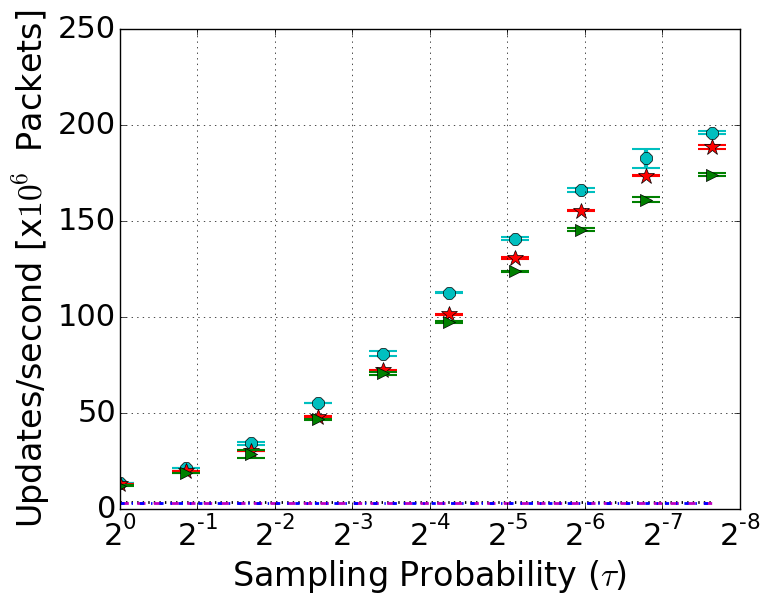}}
\subfloat[\label{HHHAccCHI16 }Backbone trace - 2 D (H=25)]{\includegraphics[width = 0.5\columnwidth]
            {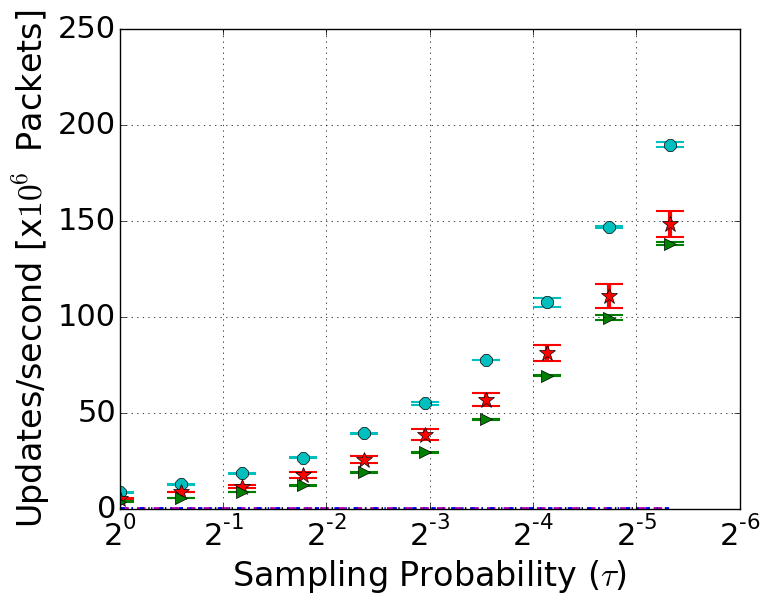}}

\caption{
\label{fig:HHHtau} Effect of the sampling probability 
on the speed of $\Hname$,  compared to the Baseline algorithm in the Backbone trace. Note that \HHHAWAII{} achieves a speedup of up to 53$\times$ in 1D and up to 273$\times$ in 2D. Results for the Edge and Datacenter traces are similar.
}
\end{figure}

\T{\HHHAWAII{} \vs interval algorithm.} Next, we compare the throughput of \HHHAWAII{} to the previously suggested RHHH~\cite{RHHH}. \HHHAWAII{} and RHHH are similar in their use of samples to increase performance. Moreover, RHHH is the fastest known interval algorithm for the HHH problem.
Our results, presented in Figure~\ref{fig:streamvswindow}, show that $\HHHAWAII{}$ is faster than RHHH for small sampling ratios. The reason lies in the implementation of the sampling. Namely, in RHHH, sampling is implemented as a geometric random variable, which is inefficient for small sampling probabilities, whereas in $\HHHAWAII{}$, it is performed using a random number table.  Still, as the sampling probability gets lower, the geometric calculation becomes more efficient, and eventually, RHHH is faster than $\HHHAWAII{}$.
This is because \HHHAWAII{} performs a Window update for most packets, while RHHH only decrements a counter.

Looking at both performance figures independently, we conclude that $\HHHAWAII{}$ achieves very high performance and is likely to incur little overheads in a virtual switch implementation in a similar manner to RHHH.  
\begin{figure}[t]
\subfloat[One dimension (H=5)] {\includegraphics[width=0.5\columnwidth]{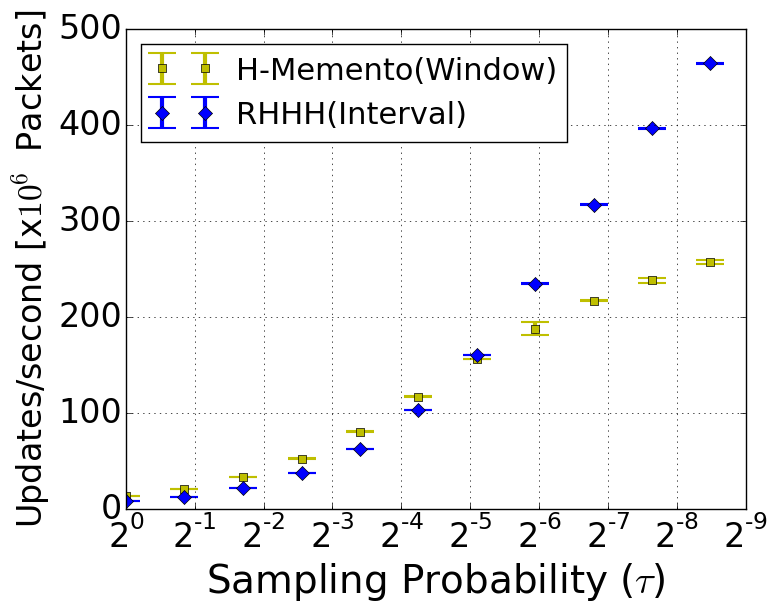}}
\subfloat[Two dimensions (H=25)] {\includegraphics[width=0.5\columnwidth]
{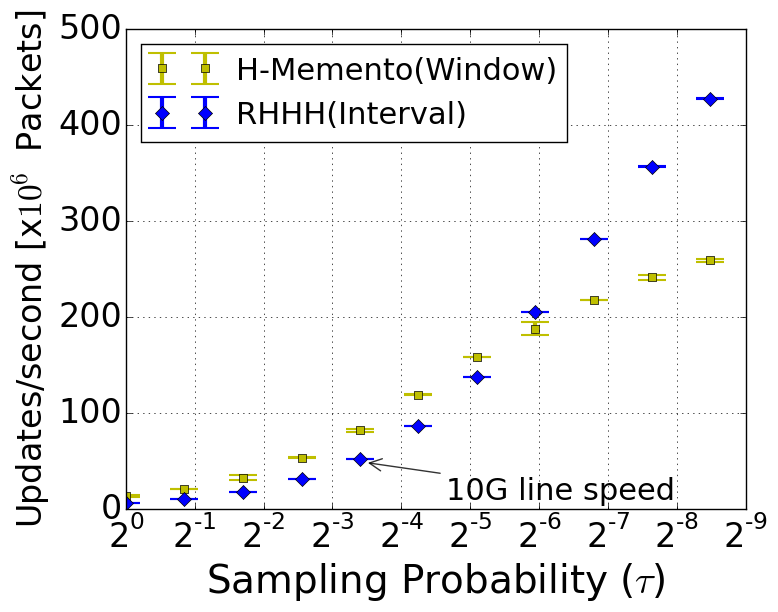}}
\caption{\label{fig:streamvswindow}
Speed comparison between RHHH (interval algorithm) and \HHHAWAII{} (window algorithm) on the Backbone dataset.  The annotated point shows the throughput of the $(\tau=1/10)$-RHHH algorithm that is reported to meet the 10G line speed using a single core~\cite{RHHH}.
That is, \HHHAWAII{} is slightly faster than RHHH in the parameter range of 10G lines.
}
\end{figure} 
\subsection{Network-Wide Evaluation}
\label{sec:eval}
This section describes our proof-of-concept system. We incorporated \HHHAWAII{} into HAProxy which provides the capability to monitor traffic from subnets, an ability which we used to implement rate limiting for subnets. Our controller periodically receives information from (in the Batch, Sample or Aggregate method) the load-balancers and uses this to perform the HHH measurement (with the D-\HHHAWAII{} algorithm). Then, the HHH output can be used as a simple threshold-based attack mitigation application where a subnet is rate-limited if its window frequency is above the threshold. 


\T{HAProxy.} We implemented and integrated our algorithms into the open-source HAProxy load-balancer (version $1.8.1$). Specifically, we leveraged and extended HAProxy's Access Control List (ACL) capabilities,  to allow the updates of our algorithms with new arriving data as well as to perform mitigation (\ie Deny or Tarpit)  when an attacker is identified. 

\T{Traffic generation.} Our goal is to obtain realistic measurements involving multiple simultaneous stateful connections such as HTTP GET and POST requests from multiple clients towards the load-balancers. To that end, we developed a tool that enables a single commodity desktop to maintain and initiate stateful HTTP GET and POST requests sourcing from multiple IP addresses. Our solution requires the cooperation of both ends (\ie the traffic generators and the load-balancer servers) for an arbitrarily large IP pool. 

It is based on the \textit{NFQUEUE} and \textit{libnetfilter}-queue Linux targets that enable the delegation of the decision on packets to a userspace software. As reported by the Apache \textit{ab load} testing tool, using a single commodity computer, we can initiate and maintain up to 30,000 stateful HTTP requests per second from different IPs without using the HTTP keep-alive feature. We are only limited by the pace at which the Linux kernel opens and closes sockets (\eg TCP timeout).

\T{Controller.} We implemented in C a test controller that communicates with the load-balancers via sockets. It holds a local HHH algorithm implementation and exchanges information with the load-balancers (\eg receives aggregations, samples, or batches). The controller then generates a global and coherent window view of the ingress traffic.

\T{Testbed.} We built a testbed involving three physical servers. The first is used for traffic generation towards the load-balancers. Specifically, we used several apache ab instances augmented with our tool to generate realistic stateful traffic from multiple IP addresses with delay and racing among different clients.  The second station holds ten autonomous instances (\ie separate processes) of HAProxy load-balancers listening on different ports for incoming requests. Finally, at the third station, we used docker to deploy  Apache server instances listening on different sockets.

\subsubsection{\HHHAWAII{}'s Accuracy}
In this experiment, we evaluate MST (denoted as \emph{Interval}),  the Baseline algorithm and \HHHAWAII{} with a single load-balancer client. 
Our goal is to monitor the last 1,000,000 HTTP requests that have entered the load-balancer. The Baseline algorithm and \HHHAWAII{} are set at $\epsilon_a =0.1\%$ and a window size of 1,000,000 requests.  The MST Interval instance is using a measurement period of 1,000,000 requests and is configured with $\epsilon_a=0.025\%$, resulting in comparable memory usage. For each new incoming HTTP request, each algorithm estimates the frequency of each of its IP prefixes.

The results are depicted in Figure \ref{fig:singlelb}. In all the traces, the Interval approach is the least accurate, while as expected, \HHHAWAII{} is slightly less accurate than the Baseline algorithm due to its use of sampling. These conclusions hold \mbox{for every prefix length and testbed workload.}



\begin{figure}[t!]
\subfloat[Backbone Trace] {\includegraphics[width=0.333\columnwidth]{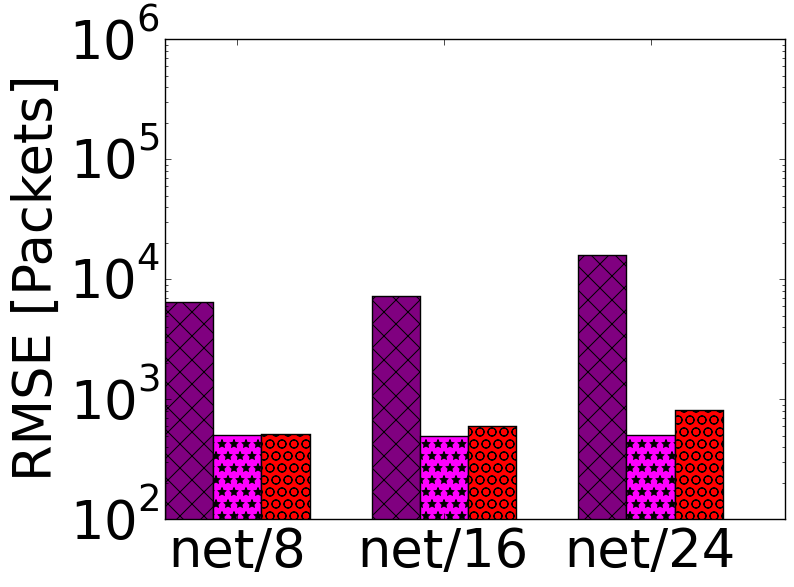}}
\subfloat[Datacenter Trace] {\includegraphics[width=0.333\columnwidth]
{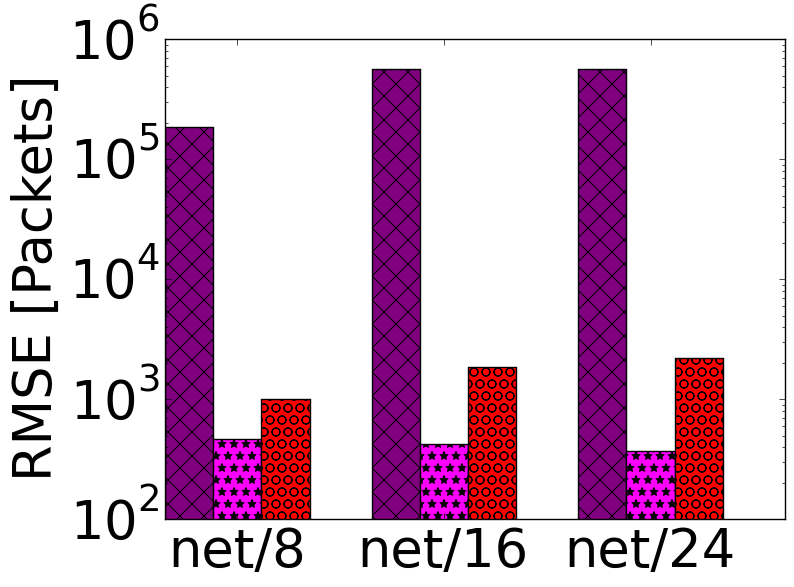}}
\subfloat[Edge Trace] {\includegraphics[width=0.333\columnwidth]
{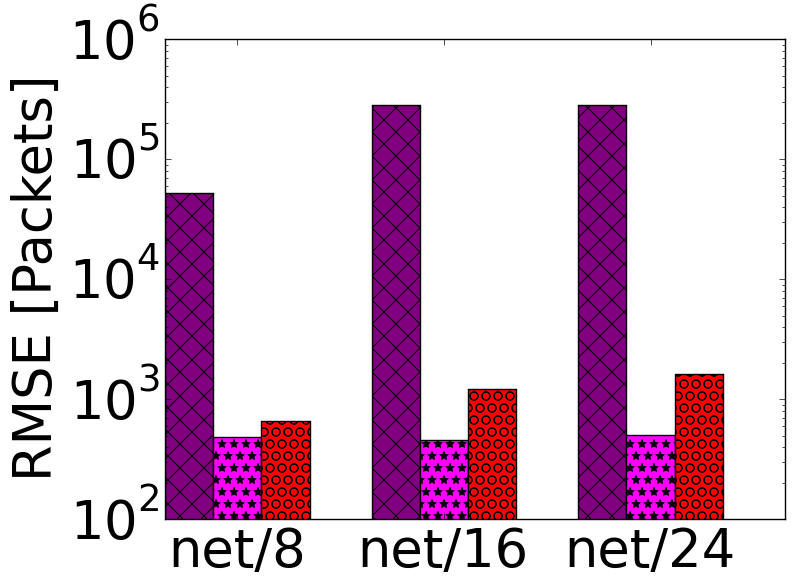}}\\{\includegraphics[width=.7\columnwidth]
{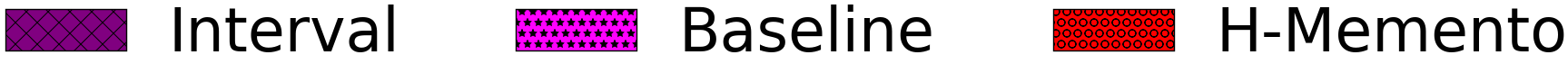}}

\caption{\label{fig:singlelb}
Comparing the error of H-Memento.  
}
\end{figure}
\begin{figure}[t]
\subfloat[Backbone Trace]
{\includegraphics[width=0.332\columnwidth]{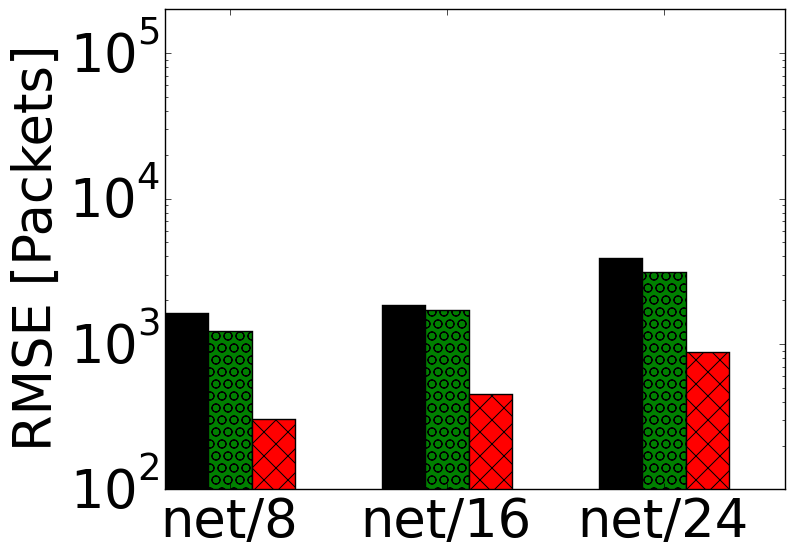}}
\subfloat[Datacenter Trace] {\includegraphics[width=0.332\columnwidth]
{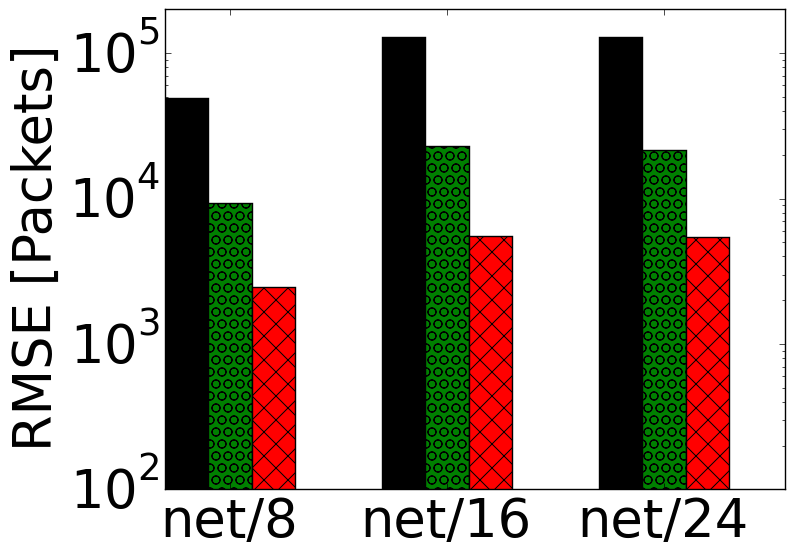}}
\subfloat[Edge Trace]
{\includegraphics[width=0.332\columnwidth]
{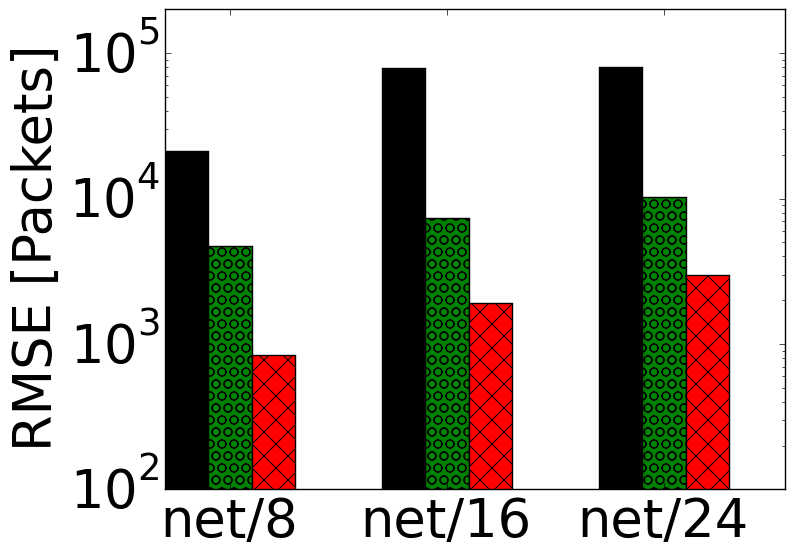}}\\{\includegraphics[width=.7\columnwidth]
{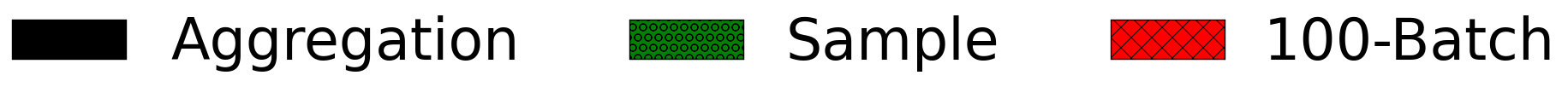}}

\caption{\label{fig:manyLBs}
Network-wide evaluation. Accuracy attained by D-\HHHAWAII{} with a bandwidth limit of 1B per ingress packet under different transmission options.
}
\end{figure}

\begin{figure*}[t!]
\subfloat[Identification over Time\label{fig:attack1}]
{\includegraphics[width=0.332\linewidth]{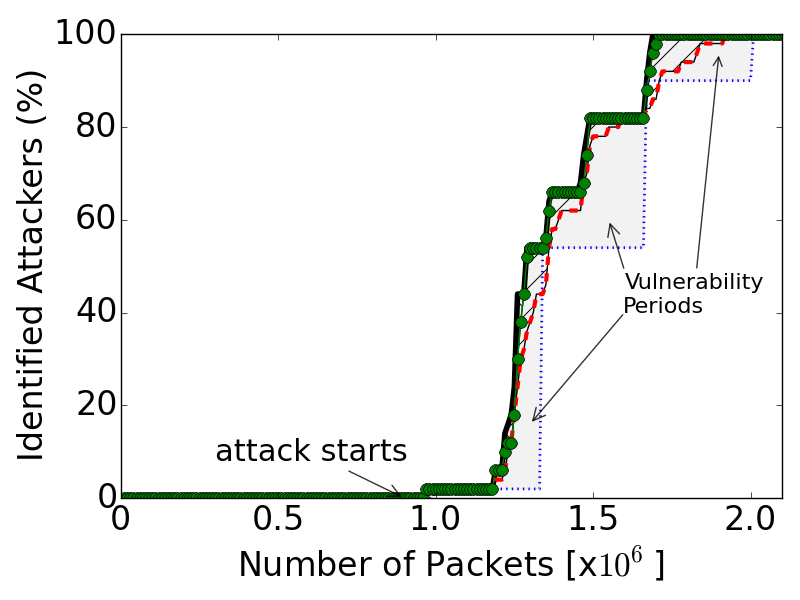}}
\subfloat[Identification (zoom)\label{fig:attack2}] {\includegraphics[width=0.332\linewidth]
{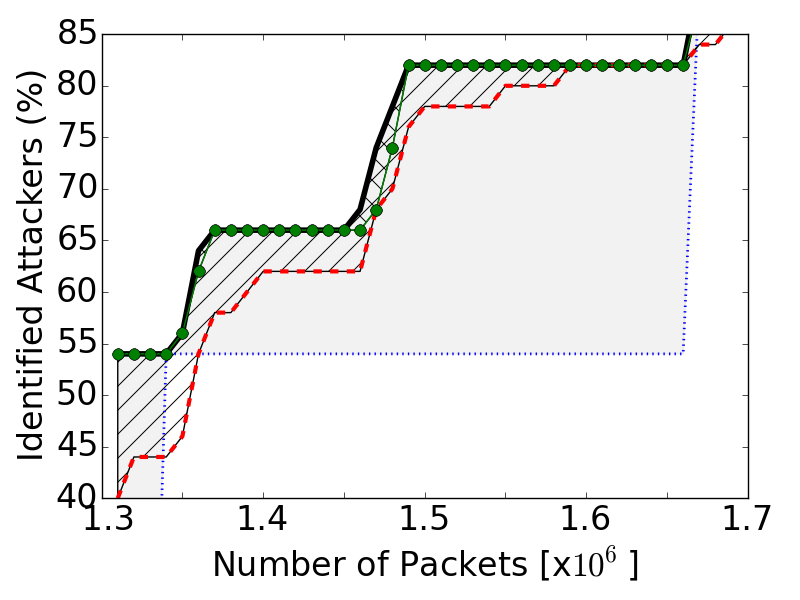}}
\subfloat[Percentage of missed attack packets\label{fig:attack3}] {\includegraphics[width=.332\linewidth]
{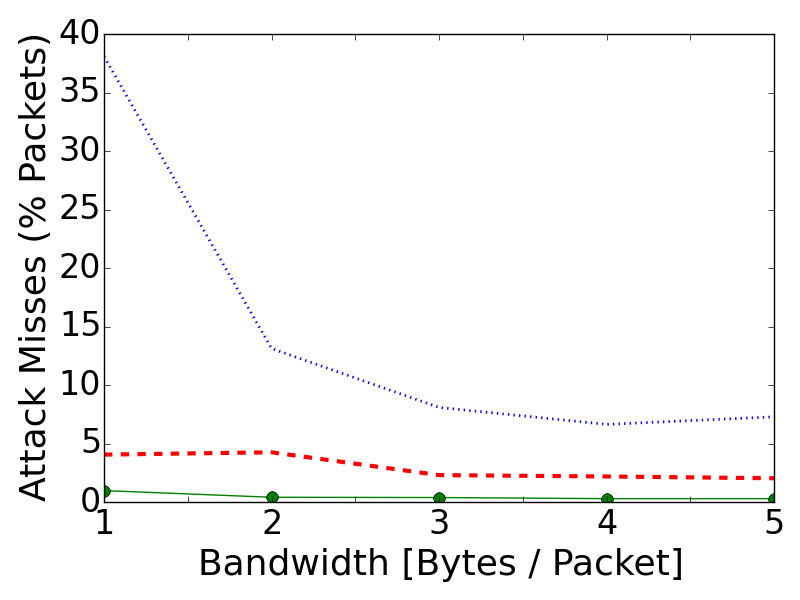}}\\
{\includegraphics[width=.5\linewidth]
{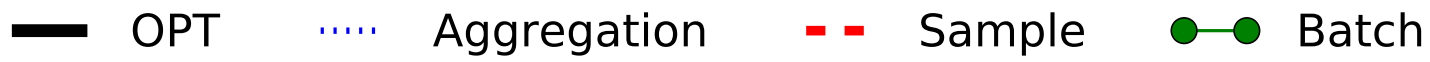}}
\caption{\label{fig:attack} HTTP flood detection experiment. 50 attacking LANs on top of the backbone trace. Comparison of the detection time of the different approaches. Our batching approach achieves near-optimal detection time.}
\end{figure*}
\subsubsection{Accuracy and Traffic Budget}
In this experiment, we generate traffic towards ten load-balancers communicating with a centralized controller that maintains a global window view of the last 1,000,000 requests that entered the system. We evaluate the three different transmission methods (Aggregation, Sample, and Batch) with the same 1-byte per packet control traffic budget.

\T{Results.} Figure \ref{fig:manyLBs} depicts the results. As indicated, the best accuracy is achieved by the Batch approach, while Sample significantly outperforms Aggregation. 
Intuitively, the Aggregation method sends the largest messages, each of which contains the full information known to the measurement point.
Its drawback is a long delay between controller updates. The Sample method has a smaller delay but utilizes the bandwidth inefficiently due to the packet header overheads. Finally, Batch has a slightly higher delay but delivers more data within the bandwidth budget, which \mbox{improves the controller's accuracy.}


\subsection{HTTP Flood Evaluation}
We now evaluate our detection system in an HTTP flood.  Our deployment consists of ten HAProxy load-balancers that serve as the entry point and direct requests to Apache servers. The HAProxy load-balancers also report to the centralized controller that discovers subnets that exceed the user-defined threshold. The bandwidth budget is set to 1-byte per packet and the window size \mbox{is $W=1$ million packets.}

\T{Traffic.} We inject flood traffic on top of the Backbone packet trace. Specifically, we select a random time at which we inject 50 randomly-picked 8-bit subnets that account for 70\% of the total traffic once the flood begins. We generate a new trace as follows. (1) We select 50 subnets by randomly choosing 8-bits for each, and (2) a random trace line in the range ($0$,$10^6$). Until that line the trace is unmodified. (3) From that line on, at each line, with probability 0.7 we add a flood line  from a uniformly picked flooding sub-network, and with probability 0.3 we  skip to the next line of the original trace. 

\T{Results.} Figure \ref{fig:attack} depicts the results. Figure~\ref{fig:attack1} and Figure~\ref{fig:attack2} show the detection speed of the flooding subnets by the three different approaches at the controller. We compare among the three approaches and additionally outline an optimal algorithm that uses an accurate window and ``knows" exactly what traffic enters the load-balancers without delay (OPT). It is notable that the Batch approach achieves near-optimal performance, and outperforms Sample and Aggregation.
Figure~\ref{fig:attack3} shows that the Batch method identifies almost all of the attack messages as is expected by our theoretical analysis.  Further, its miss rate is 37$\times$ smaller under the 1-byte per packet bandwidth budget when compared \mbox{to the ideal Aggregation method.}

%


\section{Related Work}
\label{sec:related}
\T{Heavy hitters} are an active research area on both intervals~\cite{FAST,HashPipe,SpaceSavings,RAP,DIMSUM,AndersonIMSUM,SketchVisor} and sliding windows~\cite{WCSS,HungLT10,SWHH,FAST,SWAMP}. 
HH integration in the single-device mode is an active research challenge.  For example, Sketchvisor~\cite{SketchVisor}
suggests using a lightweight fast-path measurement when the line is busy. This increases the throughput but reduces accuracy. Alternatively, HashPipe~\cite{HashPipe} adopts the interval-based Space Saving~\cite{SpaceSavings} into programmable switches. 
NetQRE~\cite{netqre} allows the network administrator to write a measurement program. The program can describe HH and HHH as well as sliding windows.  However, their algorithm is exact rather than approximate and requires a space that is linear in the window size \mbox{which is expensive for large windows.} 

\T{Hierarchical heavy hitters.} 
The HHH problem was first defined (in the Interval model) by~\cite{Cormode2003}, which also introduced the first algorithm. The problem then attracted a large body of follow-up work as well as an extension to multiple dimensions~\cite{Cormode2004,CormodeHHH,Hershberger2005,HHHMitzenmacher,Zhang:2004:OIH:1028788.1028802,MASCOTS,HHHIMC17}. MST~\cite{HHHMitzenmacher} is a conceptually simple multidimensional HHH algorithm that uses multiple independent HH instances; one instance is used for each prefix pattern. Upon a packet arrival, all instances are updated with their corresponding prefixes. The set of hierarchical heavy hitters is then calculated from the set of (plain) heavy hitters of each prefix type. The algorithm requires $O\left(\frac{H}{\epsilon}\right)$ space and $O\left(H\right)$ update time. 
MST can also operate in the sliding window model, by replacing the underlying HH algorithm with a sliding window solution~\cite{WCSS,FAST,HungAndTing}. Randomized HHH (RHH)~\cite{RHHH} is similar to MST but only updates a single HH instance.  This reduces the update complexity to a constant but requires a large amount of traffic to converge. RHHH does not naturally extend to sliding windows since each HH instance receives a slightly varying number of updates and \mbox{thus considers a different window.}

\T{Network-wide measurement.} The problem of network-wide measurement is becoming increasingly popular~\cite{RexfordNetworkwide,Sigcomm2018Networkwide, SketchVisor,FlowRadar}. A centralized controller collects data from all measurement points to form a network-wide perspective. Measurement points are placed in the network so that each packet is measured only once. The work of~\cite{HHTagging} suggests marking monitored packets which allows for more flexible measurement point placement. 

In~\cite{RexfordNetworkwide}, the controller determines a dynamic reporting threshold that allows for reduced communication overheads. It is unclear how to utilize the method in the sliding window model.  Yet, the optimization goal is very similar in essence to this work: maximize accuracy and minimize traffic overheads.  Stroboscope~\cite{Stroboscope} is another network-wide measurement system that also guarantees that the overheads adhere to a strict budget.  
FlowRadar~\cite{FlowRadar} avoids communication during the entire measurement period.
Instead, the state of each measurement point is shared at the end of the measurement. Thus, FlowRadar follows the Interval pattern, which we showed to be slow to detect new heavy hitters. 




\section{Conclusions}
\label{sec:discussion}
Our work highlights the potential benefits of sliding-window measurements to cloud operators and makes them practical for network applications. Specifically, we showed in this work that window-based measurements detect traffic changes faster, and thus enable more agile applications. Despite these benefits, sliding windows have not been used extensively, since existing window algorithms are too slow to cope with the line speed and do not provide a network-wide view. Accordingly, we introduced the Memento family of HH and HHH algorithms for both single-device and network-wide measurements. We analyzed the algorithms and extensively evaluated them on real traffic traces. Our evaluations indicate that the Memento algorithms meet the necessary speed and efficiently to provide network-wide visibility. Therefore, our work turns sliding-window HH and HHH measurements into a practical option for the next generation of network applications.

\changed{A potential drawback of existing HHH solutions, ours included, is the ability to make real-time queries. That is, while RHHH provides line-rate packet processing on streams and H-Memento provides it for sliding windows, neither allows sufficiently fast queries. Therefore, we believe that a mechanism that would allow constant-time updates for detection of changes in the hierarchical heavy hitters set would be a promising direction for future work.}

We \sout{(anonymously)} 
open-sourced the Memento algorithms and the HAProxy load-balancer extension that provides capabilities to block and rate-limit traffic from entire sub-networks (rather than from individual flows)~\cite{TRCode}. We hope that our open-source code will further facilitate sliding-window measurements in network applications.

\section*{Acknowledgments}
We thank the anonymous reviewers and our shepherd, Kenjiro Cho, for their helpful comments and suggestions.
This work was partly supported by the Hasso Plattner Institute Research School; the Zuckerman Institute; the Technion Hiroshi Fujiwara Cyber Security Research Center; and \mbox{the Israel Cyber Bureau.}

\ifdefined\fullVersion
\appendix





\section{\Hname Analysis}
\label{sec:analysis}
\T{Notations.} 
For brevity, notations are summarized in Table~\ref{tbl:notations}. 
%
%
%

\T{Main result.} 
%
%
%
%
%
The main result for \HHHAWAII{} is given by Theorem~\ref{thm:correctness}, in Section~\ref{sec:analysis}. Given the window size ($W$), the desired accuracy ($\epsilon_s$), the desired confidence ($\delta$), and the hierarchy size ($H$), Theorem~\ref{thm:correctness} provides a lower bound on the sampling probability $\tau$ that still ensures correctness. Specifically, the theorem says that \HHHAWAII{} is correct for any $\tau>\NB \triangleq \NBound$. The symbol $Z$ is a parameter that depends on $\delta$, and satisfies $Z<4$ for any $\delta> 10^{-6}$.

That is, we prove that the HHH set returned by \HHHAWAII{} satisfies the accuracy and coverage properties. Section~\ref{sec:analSamples} shows the correctness of \HHAWAII{} and the accuracy property of \HHHAWAII{}. We then show Coverage in Section~\ref{anal:randHHH}. Finally, Section~\ref{sec:RHHH-prop} shows that \HHHAWAII{} solves {\sc$(\delta,\epsilon,\theta)-$approximate windowed HHH}.

We model the update procedure of \HHHAWAII{} as a balls and bins experiment where we first select one out of $H$ prefixes and then update that prefix with probability $\tau$. For simplicity, we assume that $\frac{H}{\tau}\in\mathbb N$.
Thus, we have $V\triangleq\frac{H}{\tau}$ bins and $W$ balls.  Upon a packet arrival, we place a ball in one of the bins; if the bin is one of the first $H$, we perform a full update for the sampled prefix type, and otherwise we perform a window update.  Definition~\ref{def:Xi} formulates this model. Alternatively, \HHAWAII{} is modeled as the degenerate case where $|H|=1$ and thus we update the fully specified~prefix.

\begin{definition}
\label{def:Xi}
For each bin ($i$) and set of packets ($K$), denote by $X^{K}_i$ the number of balls (from $K$) in bin $i$. When the set $K$ contains all packets, we use the notation $X_i$.
\end{definition}

We require confidence intervals for any $X_i$ and a set $K$. However, the $X_i$'s are correlated as $\sum\nolimits_{i=1}^{V} {{X_i}}  = W$ and therefore we use the technique of Poisson approximation. It enables us to compute confidence intervals for \emph{independent} Poisson variables $\set{Y_i}$ and convert back to the balls and bins case.

Formally, let $Y_1^K,...,Y_{ V }^K\sim Poisson\left( {\frac{K}{V}} \right)$, be \textbf{independent} variables representing the number of balls in each bin.

We now use Lemma~\ref{lemma:rare} to get intervals for the $X_i$'s.
\begin{lemma}[Corollary 5.11, page 103 of~\cite{Mitzenmacher:2005:PCR:1076315}]
    \label{lemma:rare}
    Let $\mathfrak E$ be an event whose probability monotonically increases with the number of balls. If the probability of $\mathfrak E$ is $p$ in the Poisson case then it is at most $2p$ in the exact case.
\end{lemma}

\subsection{Accuracy Analysis}
\label{sec:analSamples}
To prove accuracy, we show that, for every prefix ($p$): \\ $\Pr \left( {\left| {{f^W_p} - \widehat {{f^W_p}}} \right| \le \varepsilon W} \right) \ge 1 - \delta.$
We have multiple sources of error and thus we first quantify the sampling error. Let $Y^{p}_i$ be the Poisson variable corresponding to a prefix $p$.
That is, the set $K$ contains all the packets that are generalized by $p$. Therefore: $E(Y^{p}_i) = \frac{f^W_p}{V} .$

We need to show that: \small $\Pr\left(\left|Y^{p}_i-E(Y^{p}_i)\right|\le \epsilon_s \frac{W}{V}\right)\ge 1-\delta $.
\normalsize
Luckily, there are many ways to derive confidence intervals for Poisson variables~\cite{19WaysToPoisson}. We use Lemma~\ref{lemma:poissonConfidence}, proved in~\cite{Wmethod}.

\begin{lemma}
\label{lemma:poissonConfidence}
Let $Y$ be a Poisson random variable. Then\\
$\Pr \left( {\left| {Y - E\left( Y \right)} \right| \ge {Z_{1-\delta }}\sqrt {E\left( Y \right)} } \right) \le \delta$;
here, $Z_\alpha$ is the $z$ value that satisfies $\Phi(z)=\alpha$ and $\Phi(z)$ is the cumulative density of the normal distribution with mean $0$ and STD $1$.

\end{lemma} Lemma~\ref{lemma:poissonConfidence} lays the groundwork for Theorem~\ref{thm:pusmain}, which is the main accuracy result.
\begin{theorem}
    \label{thm:pusmain}
    If $\tau \ge \NBound$ then\\
    $\Pr \left( {\left| {{X_i}^pV - {f_p}} \right| \ge {\varepsilon _s}W} \right) \le {\delta _s}.$
\end{theorem}

\begin{proof}
We use Lemma~\ref{lemma:poissonConfidence} for $\frac{\delta_s}{2}$ and get: \\
\small$\Pr \left( \left| {{Y_i}^p - \frac{{{f^W_p}}}{V}} \right| \ge {Z_{1 - \frac{\delta _s}{2}}\sqrt {\frac{{{f^W_p}}}{V}} } \right) \le \frac{\delta_s}{2}. $ \normalsize
Since we do not know the exact value of $f^W_p$, we assert that $f^W_p \le W$ to get:
$\Pr \left( \left| {{Y_i}^p - \frac{{{f^W_p}}}{V}} \right| \ge {Z_{1 - \frac{\delta _s}{2}}\sqrt {\frac{{{W}}}{V}} } \right) \le \frac{\delta_s}{2}. $
We need error of the form: $\frac{\epsilon_s\cdot W}{V}$ and thus set:
 $\frac{\epsilon_s\cdot W}{V} = {Z_{1 - \frac{\delta _s}{2}}\sqrt {\frac{{{W}}}{V}} } ={Z_{1 - \frac{\delta _s}{2}}\sqrt {\frac{{{W\tau}}}{H}} }$ We extract $\tau$ to get:
$\tau \ge {Z_{1 - \frac{{{\delta _s}}}{2}}}\frac{H}{W}{\varepsilon_s}^{ - 2} .
$
Thus, when $\tau \ge \NBound$,  we have that:
$\Pr \left( {\left| {{Y_i}^p - \frac{{{f^{W}_p}}}{V}} \right| \ge \frac{{{\varepsilon _s}W}}{V}} \right) \le \frac{{{\delta _s}}}{2} .$
We multiply by $V$ and get: \small
$\Pr \left( {\left| {{Y_i}^pV - {f^W_p}} \right| \ge {\varepsilon _s}W} \right) \le \frac{{{\delta _s}}}{2} .$\\
\normalsize
Finally, since $Y_i^{p}$ is monotonically increasing with the number of balls ($f^W_p$), we use Lemma~\ref{lemma:rare} and conclude that

$\Pr \left( {\left| {{X_i}^pV - {f_p}} \right| \ge {\varepsilon _s}W} \right) \le {\delta _s}.$
\normalsize
\end{proof}
To reduce clatter, we denote $\NB \triangleq \NBound$.
Theorem~\ref{thm:pusmain} shows that when $\tau\ge\NB$ the sample is accurate enough.
The error of the underlying \HHAWAII{} algorithm is proportional to the number of \emph{sampled} packets. Thus, if we oversample we get a slightly worse accuracy guarantee. We compensate by allocating (slightly) more counters as explained in Corollary~\ref{cor:oversample}.

\begin{corollary}
    \label{cor:oversample}
    Consider the number of updates (from the last $W$ items) to the underlying algorithm ($X$). If $\tau\ge\NB$, then $\Pr \left( {{X} \le \frac{W}{V}\left( {1 + {\varepsilon _s}} \right)} \right) \ge 1 - {\delta _s}.$
\end{corollary}

\begin{proof}
 Theorem~\ref{thm:pusmain} yields:
$\Pr \left( {\left| {{X} - \frac{W}{V}} \right| \ge {\varepsilon _s}W} \right) \le {\delta _s}.$
Thus: $\Pr \left( {{X} \le \frac{W}{V}\left( {1 + {\varepsilon _s}} \right)} \right) \ge 1 - {\delta _s}.$
\end{proof}

Corollary~\ref{cor:oversample} means that, by allocating slightly more space to the underlying algorithm, we can compensate for possible oversampling.
Generally, we configure an algorithm ($\mathbb A$) that solves {\sc {$(\varepsilon_a, \delta_a)$ - Windowed Frequency Estimation}} with
$\varepsilon_a' \triangleq \frac{\varepsilon_a}{1+\varepsilon_s}$.
Applying Corollary~\ref{cor:oversample}, we get that, with probability $1-\delta_s$, there are at most $(1+\varepsilon_s)\frac{W}{V}$ sampled packets. Using the union bound we have that with probability $1-\delta_a-\delta_s$:
$\left| {{X^p} - \widehat {{X^p}}} \right| \le {\varepsilon _{a'}}\left( {1 + {\varepsilon _s}} \right)\frac{W}{V} = \frac{{{\varepsilon _a}\left( {1 + {\varepsilon _s}} \right)}}{{1 + {\varepsilon _s}}}\frac{W}{V} = {\varepsilon _a}\frac{W}{V}.$
For example, WCSS requires $4,000$ counters for $\epsilon_a=0.001$.
If we set $\epsilon_s = 0.001$, we now require $4004$ counters.

Hereafter, we assume that the algorithm is already configured to accommodate this problem.

\begin{theorem}
    \label{thm:PUSCombined}
    Consider an algorithm ($\mathbb{A}$) that solves the {\sc {$(\epsilon_a, \delta_a)$ -Windowed Frequency Estimation}} problem.
    If $\tau \ge \NB$, then for  $\delta \ge \delta_a + 2 \cdot \delta_s$ and $\epsilon \ge \epsilon_a + \epsilon_s$, $\mathbb{A}$ solves {\sc {$(\epsilon, \delta)$ - Windowed Frequency Estimation}}.
\end{theorem}

\begin{proof}
    We use Theorem~\ref{thm:pusmain}.
    That is,
    we have that
    \begin{equation}
    \label{eq:delta2}
    \Pr \left[ {\left| {{f^W_p} - {X_p}V} \right| \ge {\varepsilon _s}W} \right] \le {\delta _s}.
    \end{equation}

    $\mathbb{A}$ solves {\sc {$(\epsilon_a, \delta_a)$ - Windowed Frequency Estimation}} and provides us with an estimator $\widehat{X^p}$ for $X^p$ --  the number of updates for a prefix $p$.
    According to Corollary~\ref{cor:oversample}:
    {\small
    $\Pr \left( {\left| {{X^p} - \widehat {{X^p}}} \right| \le \frac{{{\varepsilon _a}W}}{V}} \right) \ge 1 - {\delta _a} - {\delta _s}.$}\normalsize~Multiplying by $V$ yields:
    \begin{equation}
    \label{eq:nodelta2}
\Pr \left( {\left| {{X^p}V - \widehat {{X^p}}V} \right| \ge {\varepsilon _a}W} \right) \le {\delta _a} + {\delta _s}.
    \end{equation}
    We need to show that: $\Pr \left( {\left| {{f^W_p} - \widehat {{X^p}}V} \right| \le \varepsilon W} \right) \ge 1 - \delta$.
    Note that: $f^W_p = E(X^p)V$ and  $\widehat{f^{W}_p} \triangleq \widehat{X^p}V$.
    Thus,
    \footnotesize
    \begin{align}
    &\Pr \left( {\left| {{f^W_p} - \widehat{f^W_p}} \right| \ge \varepsilon W} \right) = \Pr \left( {\left| {{f^W_p} - \widehat {{X^p}}V} \right| \ge \varepsilon W} \right)\notag\\
    =& \Pr \left( {\left| {{f^W_p} + \left( {{X^p}{V} - {X^p}{V}} \right) - {V}\widehat {{X^p}}} \right| \ge (\epsilon_a+\epsilon_s) W} \right)\label{eq:separation}
    \\ \le&\Pr \left( \left[{\left| {{f^W_p} - {X^p}{V}} \right| \ge {\varepsilon _s}W} \right]\vee  \left[{\left| {{X^p}{V} - \widehat {{X^p}}}{V} \right| \ge {\varepsilon _a}W}\right] \right)\notag.
    \end{align}\normalsize
    The last inequality follows from the observation that if the error of~\eqref{eq:separation} exceeds $\epsilon W$, then one of the events occurs. We bound this expression with the Union bound.
    \small
\[\begin{array}{l}
\Pr \left( {\left| {{f^W_p} - \widehat {{f^W_p}}} \right| \ge \varepsilon W} \right) \le 
\Pr \left( {\left| {{f^W_p} - {X^p}V} \right| \ge {\varepsilon _s}W} \right)\\ + \Pr \left( {\left| {{X^p}V - \widehat {{X^p}}H} \right| \ge {\varepsilon _a}W} \right)  
\le{\delta _a} + 2{\delta _s},
\end{array}\]
\normalsize
where the last inequality follows from Equations (\ref{eq:delta2}) and (\ref{eq:nodelta2}).
\end{proof}

Theorem~\ref{thm:PUSCombined} implies accuracy, as it guarantees that, with probability $1-\delta$ ,the estimated frequency of \emph{any} prefix is within $\varepsilon W$ of its real frequency. In particular, this means that the HHH prefix estimations are within $\varepsilon W$ bound as shown by Corollary~\ref{cor:accuracy}.
Furthermore, by considering the degenerate case where we always select fully specified items (\ie $H=1$ and $V=\tau^{-1}$), we conclude the correctness of \HHAWAII{}, as stated in the following~Corollary~\ref{cor:RW-HH}.
\begin{corollary}
    \label{cor:accuracy}
    If $\tau \ge \NB$, then Algorithm~\ref{alg:Skipper} satisfies the accuracy constraint for $\delta = \delta_a+2\delta_s$ and $\epsilon = \epsilon_a+\epsilon_s$.
\end{corollary}

 \begin{corollary}
     \label{cor:RW-HH}
     If
     $\tau\ge Z_{1-\frac{\delta_s}{2}}W^{-1}\epsilon_s^{-2}$ then
     \HHAWAII{} solves the {\sc {$(\epsilon, \delta)$ - Windowed Frequency Estimation}} problem for $\delta = 2 \cdot \delta_s$ and $\varepsilon = \varepsilon_a + \varepsilon_s$.
 \end{corollary}

\subsection{Coverage Analysis}
\label{anal:randHHH}
We now show that \HHHAWAII{} satisfies the coverage property (Definition~\ref{def:deltaapproxHHH}). That is, $\Pr \left( \widehat {C_{q|P}} \ge C_{q|P} \right) \ge 1-\delta.$
Conditioned frequencies are calculated differently for one and two dimensions and therefore  Section~\ref{subsec:one} shows coverage for one dimension and Section~\ref{subsec:two} for two.

\subsubsection{One Dimension}
\label{subsec:one}
The following Lemma~\ref{lemma:cp}, proved in~\cite{HHHMitzenmacher} shows an expression for ${C_{q\mid P}}$.

\begin{lemma}
    \label{lemma:cp}
    In one dimension: ${C_{q\mid P}} = {f^W_q} - \sum\nolimits_{h \in G(q|P)} {{f^W_h}} .$
    \normalsize
\end{lemma}

We use Lemma~\ref{lemma:cp} to show that the estimations of Algorithm~\ref{alg:Skipper} are conservative.

\begin{lemma}
\label{lemma:sq}
The conditioned frequency estimation of Algorithm~\ref{alg:Skipper} is:
$\widehat{C_{q|P}} = \widehat{f^W_q}^{+}-\sum\nolimits_{h \in G\left( {q|P} \right)} {\widehat{f^W_h}^- } +  2{Z_{1 - \delta }}\sqrt {W V}.$
\end{lemma}
\begin{proof}
Looking at Line~\ref{line:cp} in Algorithm~\ref{alg:Skipper}, we get that: $\widehat{C_{q|P}} = \widehat{f^W_q}^{+} + calcPred(q,P)$.
That is, we need to verify that the return value $calcPred(q,P)$ in one dimension (Algorithm~\ref{alg:randHHH}) is $\sum\nolimits_{h \in G\left( {q|P} \right)} {\widehat{f^W_h}^- }$.
Finally, the addition of $2{Z_{1 - \delta }}\sqrt {WV}$ is done in line~\ref{line:accSample}.
\end{proof}

\begin{theorem}
    \label{thm:underCP}
$\Pr \left( \widehat {C_{q|P}} \ge C_{q|P} \right) \ge 1-\delta.$
\end{theorem}

\begin{proof}
from Lemma~\ref{lemma:sq}, we get \small$\widehat {{C_{q|P}}} = \widehat {f^W_q}^ +  - \sum\limits_{h \in G\left( {q|P} \right)} {\widehat {f^W_h}^ - + 2{Z_{1-\frac{\delta }{8}}}\sqrt {WV} }$.\normalsize It is enough to show that the randomness is bounded by  $2{Z_{1-\frac{\delta }{8}}}\sqrt {WV}$ with probability~$1-\delta$ as ${\widehat {{f^W_p}}^ + } \ge {f^W_p}$ and $f^W_h \le \widehat{f^W_h}^-$.
    We denote by $K$ the set of packets that affect the calculation of $\widehat {{C_{q|P}}}$. We split $K$ into two: $K^{+}$ contains packets that increase the value of $\widehat {{C_{q|P}}}$ and $K^{-}$ contains these that decrease it.
    We use $K^{+}$ to estimate the sample error in $\widehat{ f^W_q}$ and $K^{-}$ for estimating the error in $\sum\limits_{h \in G\left( {q|P} \right)} {\widehat {f^W_h}^ -}$.

We denote by $Y^{K^+}$ the number of balls in the positive sum and use Lemma~\ref{lemma:poissonConfidence}.
$C_{q|p}$ is non-negative. Thus \footnotesize$E\left(Y^{K^-}\right)\le \frac{W}{V}$\normalsize and \footnotesize
    $\Pr \left( \left| {Y^{K^ +}  - E\left( {Y^{K^ +} } \right)} \right| \ge {Z_{1-\frac{\delta }{8}}}\sqrt \frac{W}{V} \right) \le \frac{\delta }{4}.$\normalsize
    Similarly, we use Lemma~\ref{lemma:poissonConfidence} to bound the error of $Y^{K^-}$.
    $\Pr \left( {\left| Y^{K^-}  - E\left( {{Y^{K^-} }} \right) \right| \ge {Z_{1-\frac{\delta }{8}}}\sqrt  \frac{W}{V} } \right) \le \frac{\delta }{4}.$ 
    $Y^{K^+}$ and $Y^{K^-}$ are monotonic with the number of balls. We apply Lemma~\ref{lemma:rare} and use the Union bound to conclude:
    \footnotesize
$\Pr \left( {\widehat {{C_{q|P}}} \ge {C_{q|P}}} \right)\le
 2\Pr \left( {H\left( {Y^{K^-}  + Y^{K^+} } \right)   \ge VE\left( {Y^{K^-}  + Y^{K^+} } \right)\\+ 2{Z_{1 - \frac{\delta }{8}}}\sqrt {NV} } \right)
\le 1-2\frac{\delta }{2} = 1-\delta.\qedhere
$
\end{proof}
\normalsize

\normalsize
\begin{theorem}
    Algorithm~\ref{alg:Skipper} solves {\sc$(\delta, \varepsilon, \theta)$ - Approximate windowed HHH} for  $\tau \ge \NB$, $\delta = \delta_a + 2\delta_s$, $\varepsilon = \varepsilon_s + \varepsilon_a$.
\end{theorem}
\begin{proof}
We need to prove that Algorithm~\ref{alg:Skipper} satisfies both accuracy and coverage. Corollary~\ref{cor:accuracy} shows accuracy, while Theorem~\ref{thm:underCP} says that: $\Pr \left(\widehat{C_{q|P}}\ge{C_{q|P}}\right)\ge 1-\delta$.

Consider a prefix $q$ such that $q\notin P$, where $P$ is the set of HHH.  We know that $\widehat{C_{q|P}} <\theta W$ because otherwise $q$ would have been an HHH prefix. Thus, with probability $1-\delta$, we get: $C_{q|P} < \widehat{C_{q|P}}<\theta W$, which implies that  $\Pr \left( {{C_{q|P}} < \theta W} \right) \ge 1 - \delta$ and hence Algorithm~\ref{alg:Skipper} satisfies coverage as well.
\end{proof}

\subsubsection{Two Dimensions}
\label{subsec:two}
Next, we show that \HHHAWAII{} is correct for two dimensions. To do so, we require the notion of $G(q|P)$ which contains the closest prefixes to $q$ from the set $P$. Definition~\ref{def:bestG} formalizes this. \begin{definition}[Best generalization]
    \label{def:bestG}
    Define $G(q|P)$  as the set $\left\{ {p:p \in P,p \prec q,\neg \exists p' \in P:q \prec p' \prec p} \right\}$.
    Intuitively, $G(q|P)$ is the set of prefixes that are best generalized by $q$.
    That is, $q$ does not generalize any prefix that generalizes one of the prefixes in $G(q|P)$.
\end{definition}
We start with Lemma~\ref{lemma:cp2d} that quantifies the conditioned frequency in two dimensions. The proof appears in~\cite{HHHMitzenmacher}.

\begin{lemma}
    \label{lemma:cp2d}
In two dimensions,  ${C_{q|P}} = {f^W_q} - \sum\limits_{h \in G\left( {q|P} \right)} {{f^W_h}}  + \sum\limits_{h,h' \in G\left( {q|P} \right)} {{f^W_{{\rm{glb}}\left( {h,h'} \right)}}^{}} $ 
    \normalsize
\end{lemma}
Next, Lemma~\ref{lemma:algCF2D} formalizes the expression Algorithm~\ref{alg:Skipper} uses to calculate conditioned frequencies in two dimensions.
\begin{lemma}
    \label{lemma:algCF2D}
    In two dimensions, Algorithm~\ref{alg:Skipper} calculates conditioned frequency in the following manner: \footnotesize
    \[\widehat {{C_{q|P}}} = \widehat {f^W_q}^ +  - \sum\limits_{h \in G\left( {q|P} \right)} {\widehat {f^W_h}^ - }  + \sum\limits_{h,h' \in G\left( {q|P} \right)} {\widehat {f^W_{{\rm{glb}}\left( {h,h'} \right)}}^ + }  + 2{Z_{1 - \frac{\delta }{8}}}\sqrt {WV} .\]
    \normalsize
\end{lemma}
\begin{proof}
The proof follows from Algorithm~\ref{alg:Skipper}.
Line~\ref{line:cp} adds: $\widehat{f_q^{+}}$ while Line~\ref{line:accSample} is responsible for the last element ($2{Z_{1 - \frac{\delta }{8}}}\sqrt {WV}$).
The rest stems from the calcPredecessors method in Algorithm~\ref{alg:randHHH2D}.
\end{proof}

Theorem~\ref{thm:conservativeCP} lays the groundwork for coverage.
\begin{theorem}
    \label{thm:conservativeCP}
    $\Pr \left( {\widehat {{C_{q|P}}} \ge {C_{q|P}}} \right) \ge 1 - \delta .$
\end{theorem}
\begin{proof}
    Observe Lemma~\ref{lemma:cp2d} and notice that if there is no sampling error: \small$\widehat {f^W_q}^ +  - \sum\limits_{h \in G\left( {q|P} \right)} {\widehat {f^W_h}^ - }  + \sum\limits_{h,h' \in G\left( {q|P} \right)} \widehat {f^W}_{{\rm{glb}}\left( {h,h'} \right)}^{+}$\\ \normalsize
    is a conservative estimate. Thus, we now show that this error is less than $2{Z_{1 - \frac{\delta }{8}}}\sqrt {WV}$ with probability $1-\delta$.

    We denote by $K$ the packets that affect $C_{q|P}$ and since the expression of $\widehat{C_{q|P}}$ is not monotonic. As before, we split in two: $K^{+}$ are packets that increase  $\widehat{C_{q|P}}$, while and  $K^{-}$ decrease it.
    Similarly, we denote by $\{Y_i^K\}$  the number of packets from $K$ in bin $i$ of the Poisson model.
    We also denote the random variable $Y^{K^+}$ that counts \emph{how many} balls from $K$ had increased $\widehat{C_{q|P}}$.
      Lemma~\ref{lemma:poissonConfidence} binds  $Y^{K^+}$ in the following manner:
     \small
$\Pr \left( \left| {Y^{K^+}  - E\left( {Y^{K^+}} \right)} \right| \ge {Z_{1-\frac{\delta }{8}}}\sqrt \frac{W}{V} \right) \le \frac{\delta }{4}.$ \\
    Similarly, we denote by  $Y^{K^-}$ the number of packets from $K$ with negative impact on  $\widehat{C_{q|P}}$. Using Lemma~\ref{lemma:poissonConfidence} results in:
    \small
$\Pr \left( \left| {Y^{K^-}  - E\left( {Y^{K^-} } \right)} \right| \ge {Z_{1-\frac{\delta }{8}}}\sqrt \frac{W}{V } \right) \le \frac{\delta }{4}.$
\normalsize
$Y^{K^+}$ and $Y^{K^-}$ are monotonic with the number of balls.
Thus, we apply Lemma~\ref{lemma:rare} and conclude that:
\footnotesize
$
\Pr \left( {\widehat {{C_{q|P}}} \ge {C_{q|P}}} \right) \le \\
2\Pr \left( {V\left( {Y^{K^-}  + Y^{K^+} } \right) \ge \left( {VE\left( {Y^{K^-}  + Y^{K^+}} \right) + 2{Z_{1 - \frac{\delta }{8}}}\sqrt {WV} } \right)} \right)  \\
\le 1 - 2\frac{\delta }{2} = 1 - \delta$ ,
\normalsize completing the proof.
\end{proof}

\subsubsection{Putting It All Together}
We can now prove the coverage property for one and two dimensions.
\begin{corollary}
    \label{cor:coverage}
    If $\tau>\NB$ then \HHHAWAII{} satisfies coverage.
    That is, let $P$ be the HHH set of \HHHAWAII{}; given a prefix $q \notin P$ we have  $\Pr \left(C_{q|P}<\theta W\right) >1-\delta.$
\end{corollary}
\begin{proof}
    Theorem~\ref{thm:underCP} shows coverage in one dimension and
    Theorem~\ref{thm:conservativeCP} in two. These theorems guarantee that: $\Pr \left( C_{q|P}<\widehat{C_{q|P}}\right) > 1-\delta$.
    Let $q \notin P$, which means that $\widehat{C_{q|P}}<\theta W$. However, with probability $1-\delta$, $C_{q|P}<\widehat{C_{q|P}}<\theta W$ and therefore $C_{q|P} <\theta W$ as well.
\end{proof}

\subsection{\HHAWAII{} and \HHHAWAII{} Analysis}
\label{sec:RHHH-prop}
Theorem~\ref{thm:correctness} is our main result for \HHHAWAII{} and shows that \HHHAWAII{} is correct. Its proof follows from Corollary~\ref{cor:accuracy} and Corollary~\ref{cor:coverage} which imply that accuracy and coverage are both satisfied.
Note that, \HHHAWAII{} is correct when $\tau>\NB \triangleq \NBound$. That is, larger windows ($W$), or larger $\epsilon_s$, allow for more aggressive sampling.

We finish with the following simple theorems that show the space and update complexity of \HHAWAII{} and \HHHAWAII{}. 
\begin{theorem}
    \label{thm:O1}
    \HHAWAII{} and \HHHAWAII{} perform updates in constant time. 
\end{theorem}

\begin{proof}
For \HHAWAII{}, each update results in a Window or Full update, both are done in constant time. \HHHAWAII{}, also selects a random prefix in constant time. \qedhere
\end{proof}

\begin{theorem}
    \label{thm:space}
    \HHHAWAII{} uses $O\left(\frac{H}{\varepsilon_a}\right)$ table entries.
\end{theorem}
\begin{proof}
\HHHAWAII{} utilizes a \HHAWAII{} instance with $O\left(\frac{H}{\varepsilon_a}\right)$ entries. The rest of the variables are insignificant.
\end{proof}

\else
\fi

\bibliographystyle{abbrv}
\bibliography{refs}

\end{document}